\UseRawInputEncoding
\documentclass[11pt,reqno]{amsart}
\usepackage{}
\usepackage{mathrsfs}
\usepackage{amsfonts}
\allowdisplaybreaks[4]
\usepackage{graphicx} 
\usepackage{epsfig}
\usepackage{amssymb}
\usepackage{amsmath,xypic}
\usepackage{cite,color}
\newtheorem{theorem}{Theorem}

\newtheorem{proposition}[theorem]{Proposition}
\newtheorem{corollary}[theorem]{Corollary}
\newtheorem{conjecture}[theorem]{Conjecture}
\newtheorem{lemma}[theorem]{Lemma}
\newcommand{\be}{\begin{equation}}
\newcommand{\ee}{\end{equation}}
\newcommand{\bea}{\begin{eqnarray}}
\newcommand{\eea}{\end{eqnarray}}
\newcommand{\ba}{\begin{array}}
	\newcommand{\ea}{\end{array}}
\newcommand{\bean}{\begin{eqnarray*}}
	\newcommand{\eean}{\end{eqnarray*}}

\newcommand{\pa}{\partial}

\begin{document}
\title{CKP hierarchy and free Bosons}
\author{Yi Yang, Lumin Geng, Jipeng Cheng$^*$}
\dedicatory { School of Mathematics, China University of
Mining and Technology, \\Xuzhou, Jiangsu 221116, P.\ R.\ China}
\thanks{*Corresponding author. Email: chengjp@cumt.edu.cn.}
\begin{abstract}
In this paper, free Bosons are used to study some integrable properties of the CKP hierarchy, from the aspects of tau functions.
Firstly, the modified CKP hierarchy is constructed by using free Bosons and the corresponding Lax structure is given. Then the constrained CKP hierarchy is found to be related with the modified CKP hierarchy, and the corresponding solutions are derived by using free Bosons. Next by using the relations between the Darboux transformations and the squared eigenfunction symmetries, we express the Darboux transformations of the CKP hierarchy in terms of free Bosons, by which one can better understand the essential of the CKP Darboux transformations. In particular, the additional symmetries of the CKP hierarchy can be viewed as the infinitesimal generator of the CKP Darboux transformations. Based upon these results, we finally obtain the actions of the CKP additional symmetries on the CKP tau functions constructed by free Bosons.

\end{abstract}
\maketitle
\section{Introduction}
The CKP hierarchy \cite{Jimbo,Date19814} is an important sub-hierarchy of the KP hierarchy,
corresponding to the infinite Lie algebra $c_\infty$ (see Section 2 for more details). In recent years, the CKP hierarchy has attracted many researches, such as \cite{Anguelova2017,Anguelova2018, Wu2013,Liu2017, Li2020,feng,fuwei,chang,Cheng2014,licx,tian,Krichever,Geng2019,van2012}. In the Lax formulation, the CKP hierarchy is viewed as the reduction of the KP hierarchy, by adding condition $L^*=-L$ to the Lax operator $L$ expressed by pseudo-differential operators\cite{Date19814}. While in terms of wave function $w(t,z)=(1+O(z^{-1}))e^{\xi(t,z)}$ with $t=(t_1=x,t_3,t_5,\cdots)$ and  $\xi(t,z)=\sum_{k\in{1+2\mathbb{Z}_{\geq 0}}}t_kz^k$, the CKP hierarchy can be determined by the following bilinear equation \cite{Date19814}
\begin{align*}
{\rm Res}_{z}w(t,z)w(t',-z)=0,
\end{align*}
where ${\rm Res}_z \sum_i a_i z^i=a_{-1}$.
In the aspect of tau functions, CKP hierarchy owns one tau function $\tau_{\rm KP}$ inherited from KP hierarchy, which is related with the wave function $w(t,z)$ in the way below,
$$w(t,z)=\frac{\tau_{\rm KP}(x-z^{-1},-z^{-2}/2,t_3-z^{-3}/3,-z^{-4}/4,\cdots)}
{\tau_{\rm KP}(\mathfrak{t})|_{t_2=t_4=\cdots=0}}e^{\xi(t,z)}.$$
Here $\mathfrak{t}=(t_1,t_2,t_3,\cdots)$.
In \cite{Krichever}, Krichever and Zabrodin proved that $\tau_{\rm KP}|_{t_2=t_4=\cdots=0}$ is a solution of the CKP hierarchy if and only if
$$\pa_{t_2}\log \tau_{\rm KP}|_{t_2=t_4=\cdots=0}=0.$$
But the time flow $(t_2,t_4,\cdots)$ parts in $\tau_{\rm KP}$ will bring much difficulty when discussing  the CKP hierarchy, it is hoped that CKP hierarchy should have one single tau function  of its own just like the BKP hierarchy \cite{DJKM}.

In the orignal paper by Date {\it et al}\cite{Date19814}, they suggested that the CKP tau function can be constructed by using free Bosons $\phi_j$ ($j\in 1/2+\mathbb{Z}$) satisfying $$\phi_i\phi_j-\phi_j\phi_i=(-1)^{j-\frac{1}{2}}\delta_{i,-j},$$
and twisted Heisenberg algebra generated by $H_n\ (n\in2\mathbb{Z}+1)$, where $$\sum_{n\in1+2\mathbb{Z}}H_nz^{-n-1}=-\frac{1}{2}:\phi(-z)\phi(z):$$
with $\phi(z)=\sum_{j\in \mathbb{Z}}\phi_{j+\frac{1}{2}}z^{j}$ is free Bosonic field and $:\phi(-z)\phi(z):=\phi(-z)\phi(z)-\langle 0|\phi(-z)\phi(z)|0\rangle$ with the vacuum $|0\rangle$ and $\langle 0|$ defined by
\begin{eqnarray*}
	\phi_{-j}|0\rangle=0,
	\ \langle 0|\phi_j=0,\ \text {if $j>0$}.
\end{eqnarray*}
In \cite{Date19814}, Date {\it et al} obtained the result
$$\tau_{\rm KP}(t_1,0,t_3,0,\cdots)=\langle 0| e^{H(t)}g|0 \rangle^{-2},$$
where $H(t)=
\sum_{0<i\in 1+2\mathbb{Z}}  t_i H_{i}$ and $g$ is the group element corresponding to $c_{\infty}$.
Based upon these, van de Leur {\it et al}\cite{van2012} developed Date {\it et al}'s thought and introduced $H_j$ ($j\in 1/2+\mathbb{Z}$) such that
\begin{align*}
\theta(z)=2\sum_{j\in 1/2+\mathbb{Z}}H_j z^{-j-1/2}=V_-(z)^{-1}\phi(z)V_+(z)^{-1},
\end{align*}
where
$V_{\pm}(z)=\exp\sum_{\pm k>0,\rm odd}\frac{2}{k}H_kz^{-k} $. These (anti)commutation relations of all $H$'s can be written into
 \begin{gather*}
[H_i, H_j]_s = \frac{j}{2}(-1)^{[j-\frac 12]}\delta_{i,-j} ,
 \end{gather*}
where the notation $[\ {,}\ ]_s$ means the supercommutator, while
$[i]$ denotes the largest integer less than a real number $i$. They constructed one tau function $$\tau(\mathbf{t})=\langle
0|e^{H(t)}e^{\chi(t_\mathrm{odd})}g|0\rangle,$$
where $\chi(t_\mathrm{odd})=
\sum_{0< i\in\frac{1}{2}+\mathbb{Z}}  t_i H_{i}$,  $\mathbf{t}=(t,t_{\rm odd})$ and Grassmannian variables $t_{\rm odd}=(t_{\frac{1}{2}},t_{\frac{3}{2}},\cdots)$, satisfying
$t_it_j-(-1)^{4ij}t_{j}t_i=0$ with $i,j\in 1+2\mathbb{Z}$ or $1/2+\mathbb{Z}$. The corresponding relation with wave function $w(t,z)$ can be found in \cite{van2012}(see Section 2 for more details).

Later by using its Hamiltonian desities, Chang and Wu in \cite{Wu2013} obtained one single tau function $\tau_{\rm CW}$ such that,
\begin{align}
	w(t,z)=(2z)^{-1/2}\sqrt{\pa_x\varphi(t,z)^2},\quad \varphi(t,z)=e^{\xi(t,z)}\frac{\tau_{\rm CW}(t-2[z^{-1}])}{\tau_{\rm CW}(t)},
	\label{wavetaucw}
\end{align}
where $[z^{-1}]=(z^{-1},z^{-3}/3,\cdots)$.
It is shown in \cite{Cheng2014} that $\tau_{\rm CW}$ is the square root of $\tau_{\rm KP}$, i.e.,
$$\tau_{\rm CW}(t)=\sqrt{\tau_{\rm KP}(t_1,0,t_3,0,t_5,0,\cdots)}.$$
However because of the square root relation in (\ref{wavetaucw}), the CKP hierarchy can not be expressed as the usual Hirota bilinear equation by using $\tau_{CW}$. After the works above, Anguelova \cite{Anguelova2017,Anguelova2018} introduced a second bosonization of the CKP hierarchy by using an untwisted Heisenberg algebra generated by \[\sum_{n\in\mathbb{Z}}\mathcal{H}_nz^{-n-2}=\frac{1}{4z}(:\phi(z)\phi(z):-:\phi(-z)\phi(-z):).\]
But there is no corresponding relations between wave functions and the tau functions in \cite{Anguelova2017,Anguelova2018}, since it seems difficult to determine that relation by noticing the corresponding Bosonization formulas. By comparing the results of the CKP tau functions, we believe the result in \cite{van2012} is convenient to discuss the integrable properties of CKP hierarchy, for example, Darboux transformation, the squared eigenfunction symmetry and the additional symmetry, since these properties are closely related with the group action corresponding to $c_{\infty}$, and the formulation in \cite{van2012} is relatively not complicated.

The first result of this paper is the construction of the modified CKP hierarchy.
In \cite{Yangc2021,Jimbo,Willox}, the KP tau function $\tau_{\rm KP}$ and $\alpha\cdot\tau_{\rm KP}$ with $\alpha$ be a Fermionic field can be used to define the modified KP hierarchy. When considering the Bosonic field, we found that $\tau=g|0\rangle$ and $\beta\tau$ with $g\in Sp_\infty$ (see Section 2) and free Bosonic fields $\beta\in V$, can also constitute one new integrable system, which can be viewed as a subhierarchy of the modified KP hierarchy. We call this new integrable system the modified CKP hierarchy. Further it is shown that the modified CKP hierarchy is equivalent to the spectral representation of the eigenfunction \cite{Cheng2014} (The definition can be found in Section 3). Based upon this point, we obtain the Lax structure of the modified CKP hierarchy, where the representation of square eigenfunction potential (SEP) by the free Boson is essential point. Here we would like to comment that the derivation of the Lax structure from the bilinear equations is not trivial. There are many important results on this topic recently, e.g., \cite{Cheng2021,Takasaki, Liusq}. In addition, we have also considered the relation between the modified CKP hierarchy and the CKP hierarchy from the aspect of the Lax operator, which is just the Miura links between the KP and modified KP hierarchies\cite{Shaw1997}. Further, the constrained CKP hierarchy \cite{Loris1999,Hewz2007} is showed to be connected with the modified CKP hierarchy, since we found that there exist the structures of modified CKP hierarchy in bilinear equations of the constrained CKP hierarchy\cite{Geng2019}. By using free Bosons, we give the solutions for the constrained CKP hierarchy.

Another important object discussed in this paper is the CKP Darboux transformation. As an effective way to construct solutions, Darboux transformation \cite{Matveev1991,Chau1992,Oevelpa1993,Oevelrmp1993} plays a key role in integrable systems. In the KP hierarchy\cite{Chau1992,Oevelpa1993,Oevelrmp1993}, there are two basic Darboux transformations: the differential type $T_D(q)=q\pa q^{-1}$ and the integral type $T_I(r)=r^{-1} \pa^{-1}r$. The binary Darboux transformation is defined by $T(q,r)=T_I(r^{(1)})T_D(q)$, where $A^{(1)}$ means the transformed object $A$ under $T_D(q)$ or $T_I(r)$. The KP Darboux transformations can be expressed by the actions of different Fermionic fields on the tau functions\cite{Yangc2021,Willox,Willox1998,Chau1992}. But in the CKP case, the constraint $L^*=-L$ can not be guaranteed by only using $T_D$ or $T_I$. It is found in \cite{Nimmo1995,He2006,Hewz2007,Loris1999} that the appropriate Darboux transformation for the CKP hierarchy is just $T(q,q)$. In order to see the changes of the CKP tau function, it will be more convenient to express the CKP Darboux transformation by free Bosons. Different from the Fermions, the square of the Bosons are not zero and thus the exponent of Bosonic fields can not be written into finite sums. It will be shown in this paper that the Bosonic fields on the tau functions are corresponding to squared eigenfunction symmetry\cite{Oevelpa1993,Oevel1998,Aratyn,Cheng2014,Loris2001}, which is a kind of symmetry generated by the eigenfunctions and adjoint eigenfunctions. In our research experience, it seems that binary Darboux transformation $T(q,r)$ can be viewed as the exponent of squared eigenfunction symmetry. To our best knowledge, we have not found any explicit references on this point, though it should be some existed result. Therefore we will give the corresponding proofs in this paper. Based upon these facts, we will show that CKP Darboux transformation is corresponding to the action of the exponents of Bosonic fields on the tau functions.

The last research object is the additional symmetries of the CKP hierarchy\cite{He2007,Wu2013,Liu2017,Zuo2014,tian}. In this paper, we will be interested in their actions on tau functions.
The additional symmetry \cite{ASvM,Dickey1995,OS} is a kind of symmetry explicitly depending on the time variables, whose generator is the squared eigenfunction symmetry constructed by using wave functions\cite{Dickey1995}. There are usually two formulations for the additional symmetry, that is, one is the form of Sato-B\"acklund transformation of tau functions \cite{DJKM}, while another is expressed by additional time flows on the Lax operator, dressing operator and wave functions in terms of Orlov-Schulman operator\cite{CLL,FF,OS}. These two different forms can be linked by Adler-Shiota-van Moerbeke-Dickey formula\cite{ASvM,Dickey1995}. As for the CKP hierarchy, the additional symmetries are given in terms of Lax forms by He {\it et al} in \cite{He2007} without considering the corresponding actions on the CKP tau functions. As for the results in \cite{Wu2013}, the actions of additional symmetries on the tau function $\tau_{\rm CW}$ need to compute the integral of a complicated expression, which is usually very hard to obtain explicit forms except the additional symmetries of lower orders. Here in this paper, by using free Bosons, we can obtain the corresponding actions on the tau function.

The paper is organized in the way below. Firstly, the construction of the CKP hierarchy is reviewed by using free Bosons in Section 2. Then in Section 3, the bilinear equation and Lax equation of the modified CKP hierarchy are obtained. Also in Section 3, we get the solutions of the constrained CKP hierarchy. Next in Section 4, we show the relations between Darboux transformations and squared eigenfunction symmetries for the KP hierarchy, then based upon this fact, the CKP Darboux transformation is expressed in terms of free Bosons. Section 5 is devoted to the discussion of the CKP additional symmetries on the CKP tau functions constructed by free Bosons. At last, we give some conclusions and discussions in Section 6.

\section{Construction of CKP hierarchy by free Bosons}
In this section, we will review the construction of CKP hierarchy by using free Bosons. Firstly, note that CKP hierarchy is corresponding to the infinite dimensional Lie algebra $c_\infty=\bar c_\infty\oplus\mathbb{C}c$ \cite{Jimbo,Date19814,van2012,Anguelova2018},
where
$$\bar c_\infty=\left\{(a_{i,j})_{i,j\in \mathbb{Z}}|a_{i,j}=(-1)^{i+j+1}a_{-j+1,-i+1},\ a_{i,j}=0,\ \text{for $|i-j|\gg 0$}\right\}.$$
The bracket operations are given by
$$[A+\lambda c,B+\mu c]=[A,B]+\psi(A,B)c,\ \lambda,\mu\in\mathbb{C},$$
where $\psi(A,B)$ is the 2-cocycle on $\bar c_\infty$ such that
$$\psi(E_{i,j},E_{j,i})=1=-\psi(E_{j,i},E_{i,j}),\quad i\leq 0,\ j>0$$
and
$$\psi(E_{i,j},E_{k,l})=0,\quad\text{otherwise}.$$
Note that the generator of $c_\infty$ is
$$Z_{i,j}=E_{i,j}-(-1)^{i+j}E_{-j+1,-i+1}.$$
Therefore the representation $\hat\pi$ of $c_\infty$ on the Fock space $$\mathcal{F}=\text{span}\left\{(\phi_{j_1})^{m_1}(\phi_{j_2})^{m_2}\cdots
(\phi_{j_{{n-1}}})^{m_{n-1}}(\phi_{j_{n}})^{m_n}\ |\ 0\rangle|j_1>j_2>\cdots>j_n>0,\quad m_i \in \mathbb{Z}_{\geq 0}\right\}$$
 can be defined by
$$\hat\pi(Z_{i,j})=(-1)^{j-1}:\phi_{i-\frac{1}{2}}\phi_{\frac{1}{2}-j}:,\quad \hat\pi(c)=-\frac{1}{2},$$
so that
$$\hat\pi(Z_{i,j})(|u\rangle)=(-1)^j:\phi_{i-\frac{1}{2}}
\phi_{\frac{1}{2}-j}:|u\rangle,$$
where $|u\rangle\in \mathcal{F}$.

If denote
$$Sp_\infty=\{e^{\hat \pi(a_1)}e^{\hat \pi(a_2)}\cdots e^{\hat \pi(a_k)}|a_i\in c_{\infty}\},$$
then we have the lemma below
\begin{lemma}\label{gphiginverse}
Assume $g\in Sp_{\infty}$, then there exists matrix $a=(a_{mn})_{m,n\in\mathbb{Z}}$ such that
\begin{align*}
g\phi_{n-\frac{1}{2}}g^{-1}=\sum_{m\in\mathbb{Z}}a_{mn}\phi_{m-\frac{1}{2}},\quad n\in\mathbb{Z}
\end{align*}
where matrix $a$ satisfies
\begin{align*}
\sum_{l}(-1)^{l}a_{m,l}a_{-n+1,-l+1}=(-1)^n\delta_{m,n}.
\end{align*}
\end{lemma}
\begin{proof}
Without loss of generality, we can assume $g=\exp\Big(\sum_{i,j\in\mathbb{Z}}b_{ij}:\phi_{i-\frac{1}{2}}\phi_{j-\frac{1}{2}}:\Big)$.
By applying $e^ABe^{-A}=e^{adA}(B)$, one can obtain that
\begin{align*}
g\phi_{n-\frac{1}{2}}g^{-1}=\sum_{m}a_{mn}\phi_{m-\frac{1}{2}},
\end{align*}
where $a=e^{\hat b}$ and $\hat b_{mn}=(-1)^{n-1}(b_{m,-n+1}+b_{-n+1,m})$.
Note that $\hat b\in c_{\infty}$, therefore matrix $a$ belongs to the corresponding Lie group \cite{uenotakasaki}, that is the relation $\sum_{l}(-1)^{l}a_{m,l}a_{-n+1,-l+1}=(-1)^n\delta_{m,n}$.
\end{proof}
If introduce the operator $S$ on $\mathcal{F}\otimes\mathcal{F}$,
\begin{align*}
S=\sum_{k\in\frac{1}{2}+\mathbb{Z}}(-1)^{k+\frac{1}{2}}\phi_{k}\otimes\phi_{-k}
={\rm Res}_{z}\phi(z)\otimes\phi(-z).
\end{align*}
Notice that
$S(|0\rangle\otimes|0\rangle)=0.$
So if set $\tau=g|0\rangle$ with $g\in Sp_{\infty}$, then by the fact that
$[S,g\otimes g]=0$ (see \cite{Anguelova2018}),
\begin{align}
S(\tau\otimes\tau)=0,\label{ckpbilinear}
\end{align}
that is
$${\rm Res}_{z}\phi(z)\tau\otimes\phi(-z)\tau=0,$$
which is just the bilinear equation of the CKP hierarchy of free Bosonic version.

Next we will try to rewrite it into the usual forms by using the Boson-(Boson+Fermion) correspondence which is given in the proposition \cite{van2012} below.
\begin{proposition}\label{fockiso}
The correspondence
\begin{align*}
\sigma:\ \mathcal{F}\rightarrow\mathcal{B}=\mathbb{C}[t_1,t_{\frac{1}{2}}, t_3,t_{\frac{3}{2}},\cdots],\quad |u\rangle\mapsto\sigma(|u\rangle)=\langle 0|e^{H({\bf t})}|u\rangle,
\end{align*}
is an isomorphism. Further,
\begin{align*}
\sigma H_{n}\sigma^{-1}=\frac{\partial}{\partial t_{n}},\quad
\sigma H_{-n}\sigma^{-1}=(-1)^{[\frac{1}{2}-n]}\frac{n}{2}t_n,\quad n>0
\end{align*}
and
\begin{align*}
\sigma\phi(z)\sigma^{-1} =
\exp\left( \sum_{k\in 1+2\mathbb{Z}_{\geq 0}}t_kz^k \right)
\exp\left( \sum_{k\in 1+2\mathbb{Z}_{\geq 0}}
\frac{2}{k}\frac{\partial}{\partial t_k}z^{-k} \right)\nonumber\\
\hphantom{\sigma\phi(z)\sigma^{-1} =}{}\times
\sum_{0<j\in\mathbb{Z}}\left( (j-\frac{1}{2})t_{\frac{2j-1}2}(-z)^{j-1}+2\frac{\partial}{\partial
t_{\frac{2j-1}2}}z^{-j} \right).
\end{align*}
\end{proposition}
If denote $\tau({\bf t})=\langle
0|e^{H(t)}e^{\chi(t_\mathrm{odd})}g|0\rangle=\sigma(g|0\rangle )$, then there is the following expansion
$$\tau({\bf t})=\sum_{\alpha \in \rm ODP_{\rm ev}}
\tau_\alpha(t)\xi_\alpha,$$
where $\rm ODP_{\rm ev}$ is the set of all partitions in an even
number of odd parts, i.e., $\alpha=(\alpha_1,\alpha_2,\cdots,\alpha_{2k})\in \rm ODP_{\rm ev}$ satisfies $\alpha_i\in 1+2\mathbb{Z}_{\geq 0}$ and $\alpha_1>\alpha_2>\cdots>\alpha_{2k}$,
$\xi_\alpha=t_{\frac{\alpha_1}{2}}t_{\frac{\alpha_2}{2}}\cdots t_{\frac{\alpha_{2k}}{2}}$ and
\begin{eqnarray}
\tau_\alpha(t)=\frac{\partial}{\partial
t_{\frac{\alpha_n}2}} \frac{\partial}{\partial t_{\frac{\alpha_{n-1}}2}}
\cdots \frac{\partial}{\partial t_{\frac{\alpha_1}2}}\tau({\bf t})\biggr|_{t_\mathrm{odd}=0}=\langle \alpha|e^{H(t)} g|0\rangle.\label{taualpha}
\end{eqnarray}
Here we have set $\langle \alpha|=\langle
0|H_{\frac{\alpha_n}{2}}H_{\frac{\alpha_{n-1}}{2}}
\cdots
H_{\frac{\alpha_{1}}{2}}$ (e.g. $\langle 1|=\langle 0|H_{\frac{1}{2}}$). In order to rewrite (\ref{ckpbilinear}) into explicit forms, we also need
\begin{gather}
\label{sphi} \sigma(\phi(z)g|0\rangle)=\langle
0|e^{H(t)}e^{\chi(t_\mathrm{odd})}\phi(z)g|0\rangle=\sum_{\alpha \in {\rm
ODP}_\mathrm{odd}}g_\alpha(t,z)\xi_\alpha,
\end{gather}
where $\rm ODP_{\rm odd}$ is the odd partitions of odd length with distinct parts and
\begin{gather}
g_\alpha(t,z) = \frac{\partial}{\partial t_{\frac{\alpha_n}2}}
\frac{\partial}{\partial t_{\frac{\alpha_{n-1}}2}} \cdots
\frac{\partial}{\partial t_{\frac{\alpha_1}2}}
\langle 0|e^{H(t)}e^{\chi(t_\mathrm{odd})} \phi(z)g|0\rangle\biggr|_{t_\mathrm{odd}=0}\nonumber\\
\phantom{g_\alpha(t,z)}{}  = \langle
0|H_{\frac{\alpha_n}{2}}H_{\frac{\alpha_{n-1}}{2}}
\cdots
H_{\frac{\alpha_{1}}{2}}e^{H(t)} \phi(z)g|0\rangle  =\langle \alpha|e^{H(t)} \phi(z)g|0\rangle.\label{ga}
\end{gather}

After the preparation above, now we can
apply
$$\langle \alpha|e^{H(t)}\otimes\langle\beta|e^{H(t')}$$
on the bilinear equation
${\rm Res}_{z}\phi(z)\tau\otimes\phi(-z)\tau=0$,
then
$${\rm Res}_{z}g_\alpha(t,z)g_\beta(t',-z)=0.$$
Note that
\begin{align*}
&g_\alpha(t,z)=e^{\xi(t,z)}\exp\left(\sum_{j\in 1+2\mathbb{Z}_{\geq 0}}\frac{2}{j}
\frac{\pa}{\pa t_j}z^{-j}\right)\\
&\times\left(\sum_{i=1}^k(-1)^{i-1}\frac{\alpha_i}{2}
\tau_{\alpha\setminus\alpha_i}(t)(-z)^{\frac{\alpha_1-1}{2}}+2\sum_{\nu\in1+2\mathbb{Z}_{\geq 0},\nu\notin\alpha}\tau_{\alpha\cup\nu}(t)z^{-\frac{\nu+1}{2}}\right),
\end{align*}
where for $\alpha=(\alpha_1,\alpha_2,\cdots,\alpha_k)$ with $\alpha_1>\alpha_2>\cdots>\alpha_k$, $\nu\notin\alpha$ and $\alpha_i>\nu>\alpha_{i+1}$,
\begin{align*}
\alpha\cup\nu=(\alpha_1,\alpha_2,\cdot,\alpha_i,\nu,\alpha_{i+1},\cdots,\alpha_k),
\quad\alpha\setminus\alpha_i=(\alpha_1,\alpha_2,\cdot,\alpha_{i-1},\alpha_{i+1},\cdots,\alpha_k).
\end{align*}
In particular for $\alpha=(1)=1$,
\begin{eqnarray}
g_{1}(t,z)= e^{\xi(t,z)}
\exp\left(\sum_{j\in 1+2\mathbb{Z}_{\geq 0}}
\frac{2}{j}\frac{\partial}{\partial
t_j}z^{-j}\right)\left(\frac{1}{2}\tau_0(t) +2 \sum_{\nu\in 1+
2\mathbb{Z},\, \nu>1} \tau_{(\nu,1)}(t)z^{-\frac{\nu +1}{2}}\right)\label{glam},
\end{eqnarray}
where $\tau_{(\nu,1)}(t)$ means that we take $\alpha=(\nu,1)$ in $\tau_{\alpha}(t)$.

If further set the wave function corresponding to $\alpha\in {\rm ODP}_{\rm odd}$ as follows,
\begin{align}
w_\alpha(t,z)=\hat w_\alpha(t,z)z^{\frac{\alpha_1-1}{2}}e^{\xi(t,z)},\label{walpha}
\end{align}
with
\begin{align}
\hat w_\alpha(t,z)=(-z)^{-\frac{\alpha_1-1}{2}}\frac{2g_\alpha(t,z)e^{-\xi(t,z)}}
{\alpha_1\tau_{\alpha\setminus\alpha_1}(t)}=1+\mathcal{O}(z^{-1}),\label{hatwalpha}
\end{align}
then
\begin{align*}
{\rm Res}_{z}w_\alpha(t,z)w_\beta(t',-z)={\rm Res}_{z}\hat w_\alpha(t,z)\hat w_\beta(t',-z)e^{\xi(t-t',z)}z^{\frac{\alpha_1+\beta_1-2}{2}}=0.
\end{align*}
Let us take $\alpha=\beta$ and $\alpha_1=\beta_1=r$, $r$ odd. Then
\begin{align}
{\rm Res}_{z}\hat w_\alpha(t,z)\hat w_\alpha(t',-z)e^{\xi(t-t',z)}z^{r-1}=0,
\end{align}
which is just the bilinear equation of BC$_r$ hierarchy \cite{Date19814,Zuo2014,Geng2019}.

Next we will concentrate on the case of $\alpha=(1)=1$. In this case
\begin{small}
\begin{eqnarray*}
w(t,z) ={\hat
w}(t,z)e^{\xi(t,z)}=\frac{2\langle 1|e^{H(t)} \phi(z)g|0\rangle}{\langle
0|e^{H(t)}g|0\rangle},
\end{eqnarray*}
\end{small}
where ${\hat w}(t,z)=\frac{2g_{1}(t,z)e^{-\xi(t,z)} }
{\tau_{0}(t)}=1+w_1z^{-1}+w_2z^{-2}+w_3z^{-3}+\cdots$, it can be found that
\begin{eqnarray*}
w_1&=&\frac{p_1(\tilde{\partial})\tau_0(t)}{\tau_0(t)},\\
w_2&=&\frac{p_2(\tilde{\partial})\tau_0(t)}{\tau_0(t)}+4\frac{\tau_{(3,1)}(t)}{\tau_0(t)},\\
w_3&=&\frac{p_3(\tilde{\partial})\tau_0(t)}{\tau_0(t)}+4\frac{p_1(\tilde{\partial})\tau_{(3,1)}(t)}{\tau_0(t)}+4\frac{\tau_{(5,1)}(t)}{\tau_0(t)},\\
&\cdots&\nonumber\\
w_n&=&\frac{p_n(\tilde{\partial})\tau_0(t)}{\tau_0(t)}+4\sum_{l=2}^{n-1}\frac{p_{n-l}(\tilde{\partial})\tau_{(2l-1,1)}(t)}{\tau_0(t)}+4\frac{\tau_{(2n-1,1)}(t)}{\tau_0(t)},
\end{eqnarray*}
where $p_n(\tilde{\partial})$ is the Schur polynomials and $\tilde{\partial}=(2\pa_{t_1},\frac{2\pa_{t_3}}{3},\cdots)$.

\noindent{\bf Remark:}
Note that $w_1=-\pa_x\log\tau_{\rm KP}$, thus one can obtain
$\tau_{\rm KP}=\tau_0^{-2}$ up to the multiplication of some constant, which is just the result in original paper by Date {\it et al} \cite{Date19814}. While for the tau functions $\tau_{\rm CW}(t)$ introduced in \cite{Wu2013}, we can get the following relationship by comparing different expressions of $w_i$,
\begin{eqnarray*}
\tau_0(t)=\tau_{\rm CW}(t)^{-1},\quad \tau_{(3,1)}(t)=\frac{3}{4}\left(\frac{(\tau_{\rm CW}(t))_{xx}}{\tau^2_{\rm CW}(t)}-\frac{(\tau_{\rm CW}(t))^2_{x}}{\tau^3_{\rm CW}(t)}\right).
\end{eqnarray*}
Now the bilinear equation corresponding to $\alpha=(1)$ will be
\begin{eqnarray}
&&{\rm Res}_z w(t,z)w(t',-z)=0.\label{waveckpbilinear}
\end{eqnarray}

If denote
$$W=1+w_1\pa^{-1}+w_2\pa^{-2}+\cdots,\quad \text{with $\pa=\pa_x$},$$
then by the lemma below
\begin{lemma}\cite{Geng2019,DJKM} \label{newAB}
If let $A(x,\partial_{x})=\sum_{i}a_i(x)\partial_{x}^{i}$ and $B(x',\partial_{x'})=\sum_{j}b_j(x')\partial_{x'}^{j}$ be two pseduo-differential operators, then
\begin{equation*}
{\rm Res}_{z}A(x,\partial_{x})(e^{xz})\cdot B(x',\partial_{x'})(e^{-x'z})=A(x,\partial_{x})B^{*}(x,\partial_{x})
\partial_{x}(\Delta^{0}),
\end{equation*}
where $B^{*}(x,\partial_{x})=\sum_{j}(-\partial_{x})^{j}b_j(x)$, $\Delta^{0}=(x-x')^{0}$ and
\begin{eqnarray*}
\partial_{x}^{-a}(\Delta^{0})=\left\{
  \begin{array}{ll}
   0, \quad a<0, \\
    \\
    \frac{(x-x')^{a}}{a!},\quad a\geq0.
  \end{array}
\right.
\end{eqnarray*}
\end{lemma}
\noindent One can obtain
\begin{align}
W^*=W^{-1},\quad
W_ {t_n} =-\left(W\partial^nW^{-1}\right)_- W,\quad\text{$n$ odd}.\label{dressingevolution}
\end{align}

Further define the Lax operator $L$ in the way below
$$L=W\pa W^{-1}=\pa+u_1\pa^{-1}+u_2\pa^{-2}+\cdots,$$
then it can be found that
\begin{align}
L^*=-L,\quad L_{t_n}=[(L^n)_+,L],\label{ckplaxeq}
\end{align}
which is just the CKP hierarchy\cite{Date19814}. \\
{\noindent \bf Remark}: In general case of $\hat w_\alpha(t,z)=W_\alpha(e^{\xi(t,z)})$ with $\alpha_1=r$ odd, one can only obtain the BC$_r$ constraints \cite{Date19814,Zuo2014,Geng2019}
$$(W_\alpha\pa^{r-1}W_\alpha^*)_-=0,$$
which means that $P=W_\alpha\pa^{r-1}W_\alpha^*$ is a differential operator. However the evolution equation of $W_\alpha$ is hard to determine. Actually by Lemma \ref{newAB}, one can find
$$\Big((W_{\alpha,t_n}W_{\alpha}^{-1}+W_{\alpha}\pa^nW_{\alpha}^{-1})P\Big)_{<0}=0.$$
Notice that if the order of $P$ is not zero, it is usually very difficult to determine $W_{\alpha,t_n}$, since the existence of $P$ will bring many possibilities of $W_{\alpha,t_n}$. In the original paper \cite{van2012}, the authors tried to express the whole CKP system (\ref {ckpbilinear}) in the frame of super integrable hierarchies, but they only obtain the Lax operator without the Lax equation. We believe that they encountered the same problem.
\section{modified CKP hierarchy and constrained CKP hierarchy}
In this section, we first construct the modified CKP hierarchy by using free Bosons and express it in an explicit form of bilinear equation. Then by using the spectral representation of eigenfunction, the Lax structure of the modified CKP hierarchy is obtained. Based upon these results, we find there are also the structures of the modified CKP hierarchy in the constrained CKP hierarchy. We end this section with the solutions of the constrained CKP hierarchy given by free Bosons.
\subsection{The modified CKP hierarchy}
Before discussion, let us see some properties of free Bosons, which will be useful in the construction of bilinear equation for the modified CKP hierarchy. If denote
$\vec{\beta}=(\beta_m,\cdots,\beta_1)$, $\vec{\beta}\setminus\beta_i\triangleq(\beta_m,\cdots,\beta_{i+1},\beta_{i-1},
\cdots,\beta_1)$, and let $\vec{\bf n}=(n,n-1,\cdots,2,1)$, $\beta_{\vec{\bf n}}=\beta_n\beta_{n-1}\cdots\beta_2\beta_1$, then by the definition of $S$ and the commuting relation of $\phi_j$, one can obtain the lemma below.
\begin{lemma}\label{lemmas}
For $\beta_i\in V=\sum_{l\in \mathbb{Z}+1/2}\mathbb{C}\phi_l$, one has the following relations
\begin{align*}
S(1\otimes\beta_{\vec{\bf n}})=-\sum_{l=1}^n\beta_l\otimes\beta_{\vec{\bf n}\setminus l}+(1\otimes\beta_{\vec{\bf n}}) S,\quad S(\beta_{\vec{\bf n}}\otimes1)=\sum_{l=1}^n\beta_{\vec{\bf n}\setminus l}\otimes\beta_l+(\beta_{\vec{\bf n}}\otimes1) S.
\end{align*}
\end{lemma}
By this lemma and the bilinear equation of CKP hierarchy in \eqref{ckpbilinear}, one can easily obtain the proposition below
\begin{proposition}
Given $\beta\in V=\sum_{l\in \mathbb{Z}+1/2}\mathbb{C}\phi_l$ and $g\in Sp_{\infty}$,
\begin{equation}\label{mckpbilinear}
S(g|0\rangle\otimes \beta g|0\rangle)=-\beta g|0\rangle\otimes g|0\rangle.
\end{equation}
\end{proposition}
Next we will try to write \eqref{mckpbilinear} into the usual bilinear equations. For this, we can apply
$\langle1|e^{H(t)} \otimes\langle0|e^{H(t')}$ to both sides of \eqref{mckpbilinear}, then one can obtain,
\begin{equation*}
{\rm Res}_{z}\langle1|e^{H(t)}\phi(z)g|0\rangle\otimes\langle0|e^{H(t')}\phi(-z)\beta g|0\rangle=-\langle1|e^{H(t)}\beta g|0\rangle\otimes\langle0|e^{H(t')}g|0\rangle,
\end{equation*}
that is,
\begin{equation}
{\rm Res}_{z}\frac{\langle1|e^{H(t)}\phi(z)g|0\rangle}{\langle1|e^{H(t)}\beta g|0\rangle}\frac{\langle0|e^{H(t')}\phi(-z)\beta g|0\rangle}{\langle0|e^{H(t')}g|0\rangle}=-1.\label{bilinear1}
\end{equation}
Before further discussion, the lemmas below are needed.
\begin{lemma}\label{qboson}
Assume $\beta\in V=\sum_{l\in \mathbb{Z}+1/2}\mathbb{C}\phi_l$ and $g\in Sp_\infty$. If denote
\begin{align}
q(t)=\frac{2\langle 1|e^{H(t)}\beta g|0\rangle}{\langle 0|e^{H(t)}g|0\rangle},\label{qbosonexp}
\end{align}
then
$q(t)$ is the eigenfunction of the CKP hierarchy corresponding to the wave function $w(t,\lambda)=\frac{2\langle 1|e^{H(t)}\phi(\lambda) g|0\rangle}{\langle 0|e^{H(t)}g|0\rangle}$, that is, $q_{t_n}=(L^n)_{\geq 0}(q)$,
where $L=W\pa W^{-1}$ is the Lax operator and  pseudo-differential operator $W$ satisfies $w(t,\lambda)=W(e^{\xi(t,\lambda)})$.
\end{lemma}
\begin{proof}
Firstly, it is obviously that there exists $\rho(\lambda)\in\mathbb{C}((\lambda^{-1}))$ such that
$$\beta={\rm Res}_z \rho(\lambda)\phi(\lambda).$$
Therefore from the expressions of $q(t)$ and $w(t,\lambda)$ in terms of free Bosons, we can obtain the spectral representation of the eigenfunction of the CKP hierarchy \cite{Cheng2014}, that is,
\begin{align}
q(t)={\rm Res}_z \rho(\lambda)w(t,\lambda).\label{ckpspectralrepr}
\end{align}
Further note that $w(t,\lambda)_{t_n}=(L^n)_{\geq 0}(w(t,\lambda))$ which can be derived by $w(t,\lambda)=W(e^{\xi(t,\lambda)})$ and (\ref{dressingevolution}), therefore one can at last find that $q(t)$ is the CKP eigenfunction.
\end{proof}

\begin{lemma}\label{2wave}
Assuming $w(t,\lambda) =(1+\mathcal{O}(\lambda^{-1}))e^{\xi(t,\lambda)}$ and $\widetilde{w}(t,\lambda)=\sum^{\infty}_{l=1}\frac{\widetilde{w}_l}{\lambda^l}e^{-\xi(t,\lambda)}$, if $w(t,\lambda)$ and $\widetilde{w}(t,\lambda)$ satisfy
\begin{eqnarray*}
{\rm Res}_\lambda w(t,\lambda)\widetilde{w}(t',\lambda)=0
\end{eqnarray*}
then $\widetilde{w}(t,\lambda)=0$, i.e., $\widetilde{w}_1=\widetilde{w}_2=\cdots=0$.
\end{lemma}
\begin{proof}
The proof can be found in \cite{Miwa2000}.
\end{proof}
Assume the eigenfunction $f(\mathfrak{t})$ and the adjoint eigenfunction $g(\mathfrak{t})$ corresponding the Lax operator $\mathfrak{L}=\pa+\mathcal{O}(\pa^{-1})$ of the KP hierarchy, that is,
$$f_{t_n}=(\mathfrak{L}^n)_{\geq 0}(f),\quad g_{t_n}=-(\mathfrak{L}^{*n})_{\geq 0}(g),\quad n=1,2,3,\cdots,$$
then we can introduce the squared eigenfunction potential (SEP)\cite{Aratyn,Oevelpa1993,Oevel1998} $\Omega(f(\mathfrak{t}),g(\mathfrak{t}))$, determined by
\begin{eqnarray*}
\Omega(f(\mathfrak{t}),g(\mathfrak{t}))_x=f(\mathfrak{t})g(\mathfrak{t}),\quad
\Omega(f(\mathfrak{t}),g(\mathfrak{t}))_{t_n}={\rm Res}_{\pa}(\pa^{-1}g(\mathfrak{t}) (\mathfrak{L}^n)_{\geq0}f(\mathfrak{t})\pa^{-1}),
\end{eqnarray*}
where ${\rm Res}_{\pa}\sum_i a_i\pa^i=a_{-1}$.
In the CKP case by using Lemma \ref{newAB},
\begin{align*}
&{\rm Res}_\lambda w(x',\hat t,\lambda)\cdot \Omega(q(x,\hat t),w(x,\hat t,-\lambda))={\rm Res}_\lambda w(x',\hat t,\lambda)\cdot \pa_{x}^{-1}\Big(q(x,\hat t)w(x,\hat t,-z)\Big)\\
=& -(WW^*q)|_{x=x'}\Big((x'-x)^0\Big)=-q(x',\hat t),
\end{align*}
where $\hat t=(t_3,t_5,\cdots)$.
Further apply $\pa_{x'}^n$ on both sides,
\begin{align*}
&{\rm Res}_\lambda \pa_{x'}^n (w(x',\hat t,\lambda))\cdot \Omega(q(x,\hat t),w(x,\hat t,-\lambda))=-\pa_{x'}^n(q(x',\hat t)).
\end{align*}
Therefore
\begin{align}
{\rm Res}_\lambda w(t,\lambda)\cdot \Omega(q(t'),w(t',-\lambda))=-q(t),\label{qspec}
\end{align}
by using $q_{t_n}=(L^n)_{\geq 0}(q)$ and Taylor expansion of $w(t',\lambda)$ at $(x',\hat t)$.
On the other hand, one can find the following expansion
\begin{align}
\Omega(q(t),w(t,-\lambda))=(-q(t)+\mathcal{O}(\lambda^{-1}))\lambda^{-1}e^{\xi(t,-\lambda)}, \label{omegaexpansion}
\end{align}
by using $\pa^{-1}q=\sum_{l=0}^\infty(-1)^l \pa_x^l(q)\cdot\pa^{-l-1}$.
Therefore if assume $\rho(\lambda)=\mathcal{O}(\lambda^{-1})\cdot e^{-\xi(t',\lambda)}$ in the proof of Lemma \ref{qspec}, then
it can be found that
$$\rho(\lambda)=-\Omega(q(t'),w(t',-\lambda)),$$
where we have used Lemma \ref{2wave}, \eqref{ckpspectralrepr} and \eqref{qspec}.

\begin{proposition}\label{omegaqqboson}
If assume $w(t,\lambda)$ and $q(t)$ are the wave function and the eigenfunction of the CKP hierarchy respectively, then
\begin{eqnarray*}
\Omega(q(t),w(t,\lambda))=\frac{\langle0|e^{H(t)}\phi(\lambda)\beta g|0\rangle}{\langle0|e^{H(t)}g|0\rangle},
\end{eqnarray*}
where Bosonic field $\beta\in V$ is related with $q(t)$ by \eqref{qbosonexp}. In particular,
\begin{align*}
\Omega(q(t),q(t))=\frac{\langle0|e^{H(t)}\beta^2 g|0\rangle}{\langle0|e^{H(t)}g|0\rangle}.
\end{align*}
\end{proposition}
\begin{proof}
Firstly, by noticing that
\begin{eqnarray*}
q(t)^{-1}w(t,\lambda)=\frac{\langle1|e^{H(t)}\phi(\lambda)g|0\rangle}{\langle1|e^{H(t)}\beta g|0\rangle},
\end{eqnarray*}
the substraction of (\ref{bilinear1}) from (\ref{qspec}) can lead to
\begin{eqnarray*}
{\rm Res}_\lambda w(t,\lambda)\left(\Omega(q(t'),w(t',-\lambda))-\frac{\langle0|e^{H(t')}\phi(-\lambda)\beta g|0\rangle}{\langle0|e^{H(t')}g|0\rangle}\right)=0.
\end{eqnarray*}
On the other hand, it can be found that
by Proposition \ref{fockiso} and \eqref{qbosonexp}
\begin{eqnarray*}
&&\frac{\langle0|e^{H(t)}\phi(\lambda)\beta g|0\rangle}{\langle0|e^{H(t)}g|0\rangle}=\frac{1}{\tau_0(t)}\sigma(\phi(\lambda))\sigma^{-1}\sigma(\beta g|0\rangle)\biggr|_{t_\mathrm{odd}=0}\nonumber\\
&=&2\frac{1}{\tau_0(t)}e^{\xi(t,\lambda)}e^{\xi(\tilde{\pa},\lambda^{-1})}\left(\lambda^{-1}\frac{\pa}{\pa t_{\frac{1}{2}}}\langle0|e^{H(t)}e^{\chi(t_\mathrm{odd})}\beta
g|0\rangle\biggr|_{t_\mathrm{odd}=0} +\mathcal{O}(\lambda^{-2})\right)\\
&=&\Big(q(t)\lambda^{-1}+\mathcal{O}(\lambda^{-2})\Big)e^{\xi(t,\lambda)},\label{twowave}
\end{eqnarray*}
therefore $\Omega(q(t'),w(t',-\lambda))-\frac{\langle0|e^{H(t')}\phi(-\lambda)\beta g|0\rangle}{\langle0|e^{H(t')}g|0\rangle}=\mathcal{O}(\lambda^{-2})e^{-\xi(t',\lambda)}$ by \eqref{omegaexpansion}. So this proposition can be proved by Lemma \ref{2wave}.
\end{proof}
\begin{corollary}
(\ref{bilinear1}) is equivalent to the spectral representation (\ref{qspec}) of the CKP hierarchy.
\end{corollary}
After the preparation above, now we will investigate the dressing operator and the Lax equation of the modified CKP hierarchy starting from (\ref{bilinear1}).

If denote the wave function $\phi(t,z)$ and introduce the dressing operator $Z$ in the way below
\begin{eqnarray*}
&&\phi(t,z)=\frac{\langle1|e^{H(t)}\phi(z)g|0\rangle}{\langle1|e^{H(t)}\beta g|0\rangle}=Z(e^{\xi(t,z)}),
\end{eqnarray*}
where $Z=z_0+z_1\pa^{-1}+z_2\pa^{-2}+\cdots$ with $z_0=q^{-1}$,
then one can obtain the following proposition by Lemma \ref{newAB}.
\begin{proposition}\label{propmckp}
The wave function $\phi(t,z)$ satisfies
\begin{align*}
{\rm Res}_z \phi(t,z)\pa_{x'}^{-1}(z_0(t')^{-2}\phi(t',-z))=-1,
\end{align*}
which is just the bilinear equation of the modified CKP hierarchy. The dressing operator $Z$ obeys
$$Z^*=Z^{-1}z_0^{2}, \quad Z_{t_n}=-(Z\pa^{n}Z^{-1})_{\leq0}=0,\quad \text{$n$ odd}.$$
Further if set $\mathcal{L}=Z\pa Z^{-1}=\partial+v_0+v_1\partial^{-1}+v_2\partial^{-2}+\cdots$, then
\begin{align*}
\mathcal{L}^*=-e^{-2\int v_0dx}\cdot\mathcal{L}\cdot e^{2\int v_0dx},\quad
\mathcal{L}_{t_n}=[(\mathcal{L}^n)_{\geq 1},\mathcal{L}],
\end{align*}
which is the Lax structure of the modified CKP hierarchy.
\end{proposition}

Next we will give some examples of the CKP constraints $Z^*=Z^{-1}z_0^{2}$ and
$\mathcal{L}^*=-e^{-2\int v_0dx}\cdot\mathcal{L}\cdot e^{2\int v_0dx}$, which are listed as follows.
\begin{itemize}
  \item Constraints on $Z$:
\begin{eqnarray*}
z_2&=&\frac{1}{2z_0}(z_1^2-z_0z_{1x}+z_1z_{0x}),\nonumber\\
z_4&=&-\frac{1}{8z^3_0}(z_1^4+8z^3_{1}z_{0x}-12z_3z_{0x}z_0^2
-8z_3z_1z_0^2-
8z_1^2 z_{1x}z_{0}+19z_1^2z^2_{0x}
+7z^2_{1x}z^2_{0}\\
&+&12z_1z^3_{0x}
+6 z_{1xx} z_{0x}z^2_{0}-4z_1z_{0xxx}z_0^2-
6z_{1x}z_{0xx}z_0^2
+4z_{1xxx}z_0^3
\\
&+&12z_0^3z_{3x}
-12z_{0x}^2z_0z_{1x}-26z_{0x}z_0z_{1x}z_1),\\
&\vdots&\nonumber
\end{eqnarray*}
\item Constraints on $\mathcal{L}$:
\begin{align*}
v_2=&-v_1v_0-\frac{1}{2}v_{1x},\nonumber\\
v_4=&2v_1v_0^3+3v_{1x}v_0^2+\frac{3}{2}v_1v_{0xx}-3v_3v_0+\frac{3}{2}v_{1xx}v_0+3v_{1x}v_{0x}+\frac{1}{4}v_{1xxx}-\frac{3}{2}v_{3x},\nonumber\\
v_6=&-120v_{1x}v_0^2v_{0x}-60v_1v_0^2v_{0xx}-60v_1v_0v^2_{0x}-30v_{1x}v_0v_{0xx}-60v_0v_{0x}v_{1xx}-40v_{0}^4v_{1x}\\
&-20v_{1xxx}v^2_0+30v_{3x}v_{0}^2-30v_{1x}v_{0x}^2+\frac{15}{2}v_{3xxx}-\frac{1}{2}v_{1xxxxx}-5v_{3xxx}
-\frac{5}{2}v_{5x}-40v_{1xx}v^3_{0}\\
&+15v_{3xx}v_{0}+30v_{3x}v_{0x}+15v_{0xx}v_{3}-5v_{1xxxx}v_{0}-20v_{1xxx}v_{0x}
-30v_{1xx}v_{0xx}-20v_{1x}v_{0xxx}\\
&-5v_{1}v_{0xxxx}-16v_1v_0^5+20v_{3}v^3_{0}-5v_{5}v_{0},\\
&\vdots\nonumber
\end{align*}
\end{itemize}

The corresponding Lax equations can be shown below,
\begin{align*}
v_{0_{t3}}=&3v_1v_{0x}+3v_{1x}v_0+\frac{3}{2}v_{1xx}+3v_{0x}^2+3v_0v_{0xx}+v_{0xxx}+3v_0^2v_{0x}\\
v_{1_{t3}}=&3v_{3x}-3v_{1x}v_0^2+6v_1v_{1x}-6v_{1x}v_{0x}-\frac{1}{2}v_{1xxx}-3v_{1xx}v_0-3v_1v_{0xx}-6v_1v_0v_{0x}\\
v_{0_{t5}}=&-10v_1v_0v_{0xx}-30v_1v^2_0v_{0x}
+20v_1v_0v_{1x}-40v_{1x}v_0v_{0x}+50v_0v_{0x}v_{0xx}+
15v_0^2v_{1xx}-10v_1v_{1xx}\\
&+5v_0^4v_{0x}+30v_0^2v_{0x}^2-15v_{1xx}v_{0}^2+10v_{0}^3v_{0xx}
+10v_{0xxx}v_{0}^2+30v_0^3v_{1x}+30v_0^2v_1v_{0x}+\frac{15}{4}v_{1xxxx}\\
&+10v_{3xx}+20v_3v_{0x}+20v_0v_{3x}^2+20v_1v_{1xx}-10v_1v_{0x}^2+10v_{1x}^2+v_{0xxxxx}+10v_{0xx}^2+15v_{0x^3}\\
&+\frac{15}{2}v_1v_{0xxx}+\frac{35}{2}v_{1x}v_{0xx}+\frac{25}{2}v_{1xx}v_{0x}+\frac{5}{2}v_{1xxx}v_0+15v_{0x}v_{0xxx}+5v_0v_{0xxxx}+10v^2_1v_{0x}\\
&-10v_0^3v_{1x}.
\end{align*}
If eliminating $v_1$ and $v_3$, one can obtain the differential equation of $v_0$. But this equation is really too cumbersome so that there is no much meaning. So we omit it here. Next we will give another form of the modified CKP equation via the following Miura links between the modified CKP hierarchy and the CKP hierarchy. By direct computation, one can get the proposition below.
\begin{proposition}
If $q$ is an eigenfunction of the CKP hierarchy and $L$ is the corresponding Lax operator, then
\begin{align}
\mathcal{L}=q^{-1}Lq\label{relation}
\end{align}
is the Lax operator of the modified CKP hierarchy. Conversely, suppose $z_0$ is the $\partial^0$-terms in the dressing operator $Z$ of the modified CKP hierarchy, then $q\triangleq z_0^{-1}$ is the eigenfunction of the CKP hierarchy satisfying $q_{t_n}=(L^n)_{\geq 0}(q)$, and $L=q \cdot Z\pa Z^{-1}\cdot q^{-1}$ is the Lax operator of the CKP hierarchy.
\end{proposition}
\begin{proof}
In fact, this proposition can be easily proved by the formula \cite{Oevelrmp1993}
\begin{align*}
(q^{-1}Aq)_{\geq 1}=q^{-1}\cdot A_{\geq 0}\cdot q-q^{-1}\cdot A_{\geq 0}(q),
\end{align*}
where $A$ is a pseudo-differential operator.
\end{proof}
Note that $q_{t_n}=(L^n)_+(q)$ with $L=\pa+u_1\pa^{-1}+u_2\pa^{-2}+\cdots$ be the CKP Lax operator. If set $n=3$ and $n=5$, one can obtain
\begin{align*}
q_{t3}&=q_{xxx}+3u_1q_x+\frac{3}{2}u_{1x}q,\quad u_{1t3}=-\frac{1}{2}u_{1xxx}+3u_{3x}+6u_1u_{1x},\\
q_{t5}&=q_{xxxxx}+5u_1q_{xxx}+\frac{15}{2}u_{1x}q_{xx}+(5u_3+10u_1^2+5u_{1xx})q_x
+(10u_1u_{1x}+\frac{5}{4}u_{1xxx}+\frac{5}{2}u_{3x})q.
\end{align*}
By eliminating $u_1$ and $u_3$ in the relations above and setting $t_3=y$, $t_5=t$, one can obtain another form of modified CKP equation
\begin{align*}
(q_x^{-1}q_t)_x&=\frac{5}{9}\big(12q^{-3}q_x^{2}q_{xxx}-12q^{-3}q_x^2q_y-8q^{-2}q^{2}_{x}q_{xxx}+4q^{-2}q_xq_{xy}-4q^{-2}q_xq_{xxxx}+8q^{-2}q_{xx}q_{y}\\
&+3q^{-1}q_x^{-1}q_y^{2}+3q^{-1}q_x^{-1}q_{xxx}^{2}-q^{-1}q_{xxy}+q^{-1}q_{xxxxxx}+6q^{-1}q_{x}^{-1}q_{xxx}q_{y}-q^{-1}q_{x}^{-1}q_{xx}q_{xy}\\
&-q^{-1}q_{x}^{-2}(q_{xx})^2q_{xxx}+q^{-1}q_{x}q^{-2}_{x}q^{2}_{xx}q_{y}+q^{-1}q_{x}^{-1}q_{xx}q_{xxxx}+5q_x^{-2}q_{xx}q_{xxxxx}-2q^{-2}_{x}q_{xxy}q_{xx}\\
&-q_x^{-1}q_{yy}+q_x^{-1}q_{xxxy}-5q_x^{-1}q_{xxxxxx}+24q^{-5}q_x^3\int qq_ydx+16q^{-5}q_x(\int qq_ydx)^2\\
&+24q^{-5}q_{x}^3\int qq_{xxx}dx+16q^{-5}q_x(\int qq_{xxx}dx)^2-32q^{-5}\int qq_{y}\int qq_{xxx}dx\\
&+24q^{-4}q_{x}q_{xx}\int qq_{y}dx+24q^{-4}q_{x}q_{xx}\int qq_{xxx}dx-16q^{-3}q_y\int qq_ydx\\
&+16q^{-3}q_y\int qq_{xxx}dx+16q^{-3}q_{xxx}\int qq_{y}dx-16q^{-3}q_{xxx}\int qq_{xxx}dx\\
&+q^{-2}\int q_y^2dx+q^{-2}\int qq_{yy}dx-q^{-2}\int qq_{xxxy}dx-q^{-2}\int q_yq_{xxx}dx\\
&+2q^{-2}q_{x}^{-2}q_{xx}q_{xxx}\int qq_{y}dx-2q^{-2}q_{x}^{-2}q_{xx}q_{xxx}\int qq_{xxx}dx-2q^{-2}q_{x}^{-2}q_{xx}q_{y}\int qq_{y}dx\\
&+2q^{-2}q_{x}^{-2}q_{xx}q_{y}\int qq_{xxx}dx-2q^{-2}q_{x}^{-1}q_{xy}\int qq_{xxx}dx-2q^{-2}q_{x}^{-1}q_{xxxx}\int qq_{y}dx\\
&+2q^{-2}q_{x}^{-1}q_{xxxx}\int qq_{xxx}dx+2q^{-2}q^{-1}_{x}q_{xy}\int qq_ydx-q^{-1}q^{-2}_{x}q_{xx}\int qq_{yy}dx\\
&
+q^{-1}q^{2}_{x}q_{xx}\int qq_{xxxy}dx-q^{-1}q^{-2}_{x}q_{xx}\int q_y^2dx+q^{-1}q^{-2}_{x}q_{xx}\int q_yq_{xxx}dx
\big).
\end{align*}
\subsection{The constrained CKP hierarchy}
Another important object in this section is the $k$-constrained CKP hierarchy ($k$ is a fixed odd positive integer)\cite{Hewz2007,Loris1999}, defined by the following constraint on the CKP Lax operator,
\begin{equation}\label{cCKPL}
    L^{k}=\pa^{k}+\sum_{i=2}^ka_{i}\pa^{k-i}+q_1\pa^{-1}q_2
    +q_{2}\pa^{-1}q_1,
\end{equation}
where $q_1(t)$ and $q_2(t)$ are two independent eigenfunctions of the CKP hierarchy. So we can find that the system of $k$-constrained CKP hierarchy contains the data $(L,q_1,q_2)$ satisfying \eqref{cCKPL}, \eqref{ckplaxeq} and $q_{i,t_n}=(L^n)_{\geq 0}(q_i)$ with $i=1,2$.
When $k=1$, if we set $q_1=q_2=q$, then one can obtain the modified KdV equation $q_{t_3}=12q^2q_{x}+q_{xxx}$.

The bilinear equations of the constrained CKP hierarchy, which can be viewed as a special case of the ones for BC$_r$-KP hierarchy \cite{Geng2019}, is given in the proposition below.
\begin{theorem}\label{ckpequation}
The system of $k$-constrained CKP hierarchy \eqref{cCKPL} is equivalent to the following bilinear equations
\begin{align}
{\rm Res}_\lambda& \lambda^{k}w(t,\lambda)w(t',-\lambda)=q_1(t)q_2(t')+q_2(t)q_1(t'),\label{cbkpbilinearwave1}\\
{\rm Res}_\lambda&w(t,\lambda)\Omega(q_i(t'),w(t',-\lambda))=-q_i(t),\quad i=1,2.\label{cbkpbilinearwave2}
\end{align}
\end{theorem}
{\noindent\bf Remark}: Note that there is the structure of the modified CKP hierarchy in the constrained CKP hierarchy.

If denote $\varrho_i(t)=q_i(t)\tau_0(t)$ and notice that $a_i=a_i(w_1,\cdots,w_{i-1})$ and $w_i=w_i(\tau_0,\tau_{(3,1)},\cdots,\tau_{(2i-1)})$, then it can be found the system of the $k$-constrained CKP hierarchy involves $k+1$ functions, that is,
$\tau_0,\tau_{(3,1)},\cdots,\tau_{(2k-3,1)}, \varrho_1, \varrho_2$.
Next the solutions of the $k$-constrained CKP hierarchy are given in the proposition below
\begin{proposition}\label{propsolution}
If $g\in Sp_\infty$ satisfies
\begin{eqnarray*}\label{flj}
g^{-1}H_{k}g=\sum_{l,j\in\mathbb{Z}}f_{l,j}
\phi_{l-\frac{1}{2}}\phi_{j-\frac{1}{2}},
\end{eqnarray*}
where the constant $f_{l,j}$ can be decomposed into
\begin{eqnarray}
f_{l,j}=-(d_{l}e_{j}+e_{j}d_{l})\quad \text{for $j,l\geq1$},
\end{eqnarray}
then the following expressions
\begin{align*}
\tau_0(t)&=\langle 0|e^{H(t)}g|0\rangle,\quad
\tau_{(2i-1,1)}(t)=\langle 0| H_{\frac{2i-1}{2}}H_{\frac{1}{2}}e^{H(t)}g|0\rangle, \quad i=2,3,\cdots, k-1,\\
\varrho_1(t)&=2\langle1|e^{H(t)}g\sum_{l\geq1}d_l\phi_{l-\frac{1}{2}}|0\rangle,\quad
\varrho_2(t)=2\langle1|e^{H(t)}g\sum_{j\geq1}e_j\phi_{j-\frac{1}{2}}|0\rangle,
\end{align*}
satisfy the bilinear equations (\ref{cbkpbilinearwave1}).
\end{proposition}
\begin{proof}
Firstly by using Lemma \ref{gphiginverse}, there exists matrix $(a_{m,n})_{m,n\in\mathbb{Z}}$ such that
\begin{eqnarray}
\phi_{n-\frac{1}{2}}g=\sum_{n\in\mathbb{Z}}(a^{-1})_{m,n}g\phi_{m-\frac{1}{2}},\label{phig}
\end{eqnarray}
therefore
\begin{align*}
{\rm Res}_\lambda\lambda^{k}w(t,\lambda)w(t',-\lambda)=&-\sum_{n\in\mathbb{Z}}(-1)^{n}\frac{\langle 1|e^{H(t)} \phi_{n-\frac{1}{2}}g|0\rangle}{\langle
0|e^{H(t)}g|0\rangle}\frac{\langle 1|e^{H(t')} \phi_{-k-n+\frac{1}{2}}g|0\rangle}{\langle
0|e^{H(t')}g|0\rangle}\\
=&-\sum_{n\in\mathbb{Z}}\sum_{j,l\geq 1}(-1)^{n}(a^{-1})_{j,n}(a^{-1})_{l,-n-k+1}\frac{\langle 1|e^{H(t)} g\phi_{j-\frac{1}{2}}|0\rangle}{\langle
0|e^{H(t)}g|0\rangle}\frac{\langle 1|e^{H(t')} g\phi_{l-\frac{1}{2}}|0\rangle}{\langle
0|e^{H(t')}g|0\rangle}.
\end{align*}
On the other hand by (\ref{phig}) and $H_k=\frac{1}{2}\sum_{n\in\mathbb{Z}}(-1)^n\phi_{n-\frac{1}{2}}\phi_{-n-k+\frac{1}{2}}$, we find that
\begin{eqnarray*}
\sum_{n\in\mathbb{Z}}(-1)^{n}(a^{-1})_{j,n}(a^{-1})_{l,-n-k+l}=2f_{jl}=-2(d_{j}e_{l}+e_{j}d_{l}).
\end{eqnarray*}
So if set $\varrho_1(t)=2\langle1|e^{H(t)}g\sum_{l\geq1}d_l\phi_{l-\frac{1}{2}}|0\rangle$,
$\varrho_2(t)=2\langle1|e^{H(t)}g\sum_{j\geq1}e_j\phi_{j-\frac{1}{2}}|0\rangle$, one can obtain \eqref{cbkpbilinearwave1}.
\end{proof}
Let us give an example to show the existence of $g$ in Proposition \ref{propsolution}. Assume $k=3$ and $g=e^{-a\phi_{\frac{7}{2}}\phi_{\frac{9}{2}}}\in Sp_\infty$, where $a$ is a constant, then
\begin{eqnarray}
g^{-1}H_{k}g=H_{k}-a(\phi_{\frac{7}{2}}\phi_{\frac{3}{2}}+\phi_{\frac{1}{2}}\phi_{\frac{9}{2}}).
\end{eqnarray}
It can be found that there are only two values for $f_{l,j}$ with $l,j\geq1$, that is,  $f_{4,2}=-a$ and $f_{1,5}=-a$. So we can choose $d_{1}=d_4=a$, $d_2=d_5=0$ and $e_{2}=e_5=1$, $e_1=e_4=0$. Then the corresponding $\varrho_1$ and $\varrho_2$ can be given by
\begin{eqnarray*}
\varrho_1(t)=2\langle1|e^{H(t)}ga(\phi_{\frac{1}{2}}+\phi_{\frac{7}{2}})|0\rangle,\quad
\varrho_2(t)=2\langle1|e^{H(t)}g(\phi_{\frac{3}{2}}+\phi_{\frac{9}{2}})|0\rangle.
\end{eqnarray*}

\section{The Darboux transformation of the CKP hierarchy via free Bosons}
In this section, we firstly review some classical results of the squared eigenfunction symmetry and Darboux transformation for the KP hierarchy. It is found that the infinitesimal element of binary Darboux transformation is just the squared eigenfunction symmetry. Then based upon these results, the Darboux transformation of the CKP hierarchy is formulated in the language of free Bosons, which can help to better understand the essential properties of the CKP hierarchy.
\subsection{Squared eigenfunction symmetry via binary Darboux transformation}
In this subsection, we will investigate the relations of squared eigenfunction symmetries and binary Darboux transformations in the KP case. Firstly, consider the following additional flows $\pa^*_{q,r}$ on the Lax operator $\mathfrak{L}$ and the dressing operator $\mathfrak{D}$ of the KP hierarchy\cite{Oevelpa1993,Oevel1998,Aratyn},
$$\pa^*_{q,r}\mathfrak{L}=[q\pa^{-1}r,\mathfrak{L}],\quad \pa^*_{q,r}\mathfrak{D}=q\pa^{-1}r\mathfrak{D},$$
where the eigenfunction $q(\mathfrak{t})$ and the adjoint eigenfunction $r(\mathfrak{t})$ with $\mathfrak{t}=(t_1,t_2,\cdots)$ are defined in the way below,
$$q_{t_n}=\mathfrak{L}^n_{\geq k}(q),\ r_{t_n}=-\Big(\mathfrak{L}^n_{\geq k}\Big)^*(r),\quad n=1,2,3\cdots.$$
It can be proved that
$$[\pa^*_{q,r},\pa_{t_n}]=0,$$
therefore $\pa^*_{q,r}$ is a kind of symmetry flow called squared eigenfunction symmetry. The actions of $\pa^*_{q,r}$ on another eigenfunction $q'$ and adjoint eigenfunction $r'$ are given by \cite{Oevelpa1993,Aratyn}
\begin{align}
&\pa^*_{q,r}q'=q\Omega(q',r),\quad \pa^*_{q,r}r'=\Omega(q,r'),\nonumber\\
&\pa^*_{q,r}\Omega(q',r')=\Omega(q,r')\Omega(q',r),\quad \pa^*_{q,r}\tau_{\rm KP}=-\Omega(q,r)\tau_{\rm KP}.\label{paqrobj}
\end{align}

\begin{lemma}\label{paqrcomm}
Up to some constants in SEP,
\begin{align}
[\pa^*_{q,r},\pa^*_{q',r'}]\mathfrak{N}=0,\quad \mathfrak{N}\in \{\mathfrak{L},\mathfrak{D},q''(\mathfrak{t}),r''(\mathfrak{t}),\psi(\mathfrak{t},\lambda), \psi^*(\mathfrak{t},\lambda), \Omega(q'',r''),\tau_{\rm KP}(\mathfrak{t})\},
\end{align}
where $q,q',q''$ and $r,r',r''$ are the eigenfunctions and adjoint eigenfunctions, and $\psi(\mathfrak{t},\lambda)=\mathfrak{D}(e^{\eta(\mathfrak{t},\lambda)})$, $\psi^*(\mathfrak{t},\lambda)=\mathfrak{D}^{-1*}(e^{-\eta(\mathfrak{t},\lambda)})$ with $\eta(\mathfrak{t},\lambda)=\sum_{i=1}^\infty t_i\lambda^i$ are the wave function and adjoint wave function of the KP hierarchy.
\end{lemma}
\begin{proof}
The proof for the cases of $\mathfrak{N}=\mathfrak{L},\mathfrak{D},q'',r'',\psi, \psi^*$ can be found in \cite{Oevel1998}. Here we only discuss the remaining cases. Note that by (\ref{paqrobj}),
\begin{align*}
\pa^*_{q',r'}\pa^*_{q,r}\tau_{\rm KP}=-\pa^*_{q',r'}(\Omega(q,r)\tau_{\rm KP})
=-\Omega(q,r')\Omega(q',r)\tau_{\rm KP}+\Omega(q,r)\Omega(q',r')\tau_{\rm KP},
\end{align*}
therefore $[\pa^*_{q,r},\pa^*_{q',r'}]=0$ on $\tau_{\rm KP}$. As for the case of $\Omega(q'',r'')$, it can be proved by similar method.

\end{proof}

The relations of $\pa^*_{q,r}$ and the binary Darboux transformation $T(q,r)$ (see Section 1) are shown in the proposition below,
\begin{proposition}\label{onedtsesy}
Up to some constants in SEP,
$$\mathfrak{N}^{[1]}=e^{\pa^*_{q,r}}\mathfrak{N},\quad \mathfrak{N}\in \{\mathfrak{L},\mathfrak{D},q'(\mathfrak{t}),r'(\mathfrak{t}),\psi(\mathfrak{t},\lambda), \psi^*(t,\lambda), \Omega(q',r'),\tau_{\rm KP}(\mathfrak{t})\},$$
where $\mathfrak{N}^{[1]}$ means the transformed $\mathfrak{N}$ under $T(q,r)$, while $q'$ and $r'$ are the eigenfunction and adjoint eigenfunction different from $q$ and $r$.
\end{proposition}
\begin{proof}
Let us firstly consider the case of $\mathfrak{N}=q$. In fact by (\ref{paqrobj}), one can prove
\begin{align}
\pa^{*m}_{q,r}(q')=m!q\Omega(q,r)^{m-1}\Omega(q',r),\quad m\geq 1\label{paqrmq}
\end{align}
by induction on $m$. Therefore one can find that
\begin{eqnarray}
e^{\pa^*_{q,r}}(q')&=&\sum_{m=0}^{\infty}\frac{\pa^{*m}_{q,r}(q')}{m!}
=q'+\sum_{m=1}^{\infty}q\Omega(q,r)^{m-1}\Omega(q',r)=
q'+\frac{q\Omega(q',r)}{1-\Omega(q,r)}.\nonumber
\end{eqnarray}
On the other hand,
\begin{eqnarray}
q'^{[1]}=q'-\frac{q\Omega(q',r)}{\Omega(q,r)}.\label{q1dt}
\end{eqnarray}
If absorb the $-1$ into $\Omega(q,r)$ (since its definition can be up to a constant), then we can find that $e^{\pa^*_{q,r}}(q')$ coincides with $q'^{[1]}$. Similarly, one can prove the cases of $\mathfrak{N}=r,\psi,\psi^*$. While for $\mathfrak{N}=\Omega(q',r')$, it can be found that
\begin{align*}
\pa^{*m}_{q,r}\Omega(q',r')=m!\Omega(q,r)^{m-1}\Omega(q',r)\Omega(q,r'),\quad m\geq 1
\end{align*}
and
\begin{align*}
\Omega(q'^{[1]},r'^{[1]})=\Omega(q',r')-\frac{\Omega(q,r')\Omega(q',r)}{\Omega(q,r)}.
\end{align*}
The case of $\mathfrak{N}=\Omega(q',r')$ can be proved by viewing $\Omega(q,r)-1$ as $\Omega(q,r)$. The case of $\mathfrak{N}=\tau_{\rm KP}$ can be derived by using $(\pa^*_{q,r})^2(\tau_{\rm KP})=0$ (see the proof of Lemma \ref{paqrcomm}) and $\pa^*_{q,r}\tau_{\rm KP}=-\Omega(q,r)\tau_{\rm KP}$.

Now we try to prove the case of $\mathfrak{N}=\mathfrak{D}$ or $\mathfrak{L}$. Actually,
the action of $\pa^*_{q,r}$ on $\eta(\mathfrak{t},\lambda)$ vanishes. Thus
one can know that $\mathfrak{D}^{[1]}=e^{\pa^*_{q,r}}(\mathfrak{D})$
from $\psi^{[1]}=e^{\pa^*_{q,r}}(\psi)$ and $\psi^{[1]}=\mathfrak{D}^{[1]}(e^{\eta(\mathfrak{t},\lambda)})$. As for the case of
$\mathfrak{N}=\mathfrak{L}$, it can be found that
\begin{align}
e^{\pa^*_{q,r}}(\mathfrak{L})=e^{\pa^*_{q,r}}(\mathfrak{D})\cdot \pa\cdot e^{\pa^*_{q,r}}(\mathfrak{D}^{-1}),
\end{align}
which can be proved by direct computation or viewing $\mathfrak{L}$ or $\mathfrak{D}$ depending on $t^*$ corresponding to $\pa^*_{q,r}$. In this point, $e^{\pa^*_{q,r}}(\mathfrak{L}(\mathfrak{t},t^*))=\mathfrak{L}(\mathfrak{t},t^*+1)$ and $e^{\pa^*_{q,r}}(\mathfrak{D}(\mathfrak{t},t^*))=\mathfrak{D}(\mathfrak{t},t^*+1)$. Further by $e^{\pa^*_{q,r}}(\mathfrak{D})\cdot e^{\pa^*_{q,r}}(\mathfrak{D}^{-1})
=e^{\pa^*_{q,r}}(\mathfrak{D}\cdot\mathfrak{D}^{-1})=1$, one can obtain $e^{\pa^*_{q,r}}(\mathfrak{D}^{-1})=(e^{\pa^*_{q,r}}(\mathfrak{D}))^{-1}$. Therefore, one can prove $\mathfrak{L}^{[1]}=e^{\pa^*_{q,r}}(\mathfrak{L})$ by $\mathfrak{L}^{[1]}=\mathfrak{D}^{[1]}\cdot\pa\cdot\mathfrak{D}^{[1]-1}$.
\end{proof}
By Lemma \ref{paqrcomm} and the fact that $T(q,r)$ can commute with each other (that is, $T(q_2^{[1]},r_2^{[1]})T(q_1^{[0]},r_1^{[0]})
=T(q_1^{[1]},r_1^{[1]})T(q_2^{[0]},r_2^{[0]})$), one can obtain the corollary below.
\begin{corollary}\label{cordtkstep}
Under the same conditions of Proposition \ref{onedtsesy},
$$\mathfrak{N}^{[k]}=e^{\sum_{i=1}^k\pa^*_{q_i,r_i}}\mathfrak{N},\quad \mathfrak{N}\in \{\mathfrak{L},\mathfrak{D},q'(\mathfrak{t}),r'(\mathfrak{t}),\psi(\mathfrak{t},\lambda), \psi^*(\mathfrak{t},\lambda), \Omega(q',r'),\tau_{KP}(\mathfrak{t})\}.$$
\end{corollary}

\subsection{CKP Darboux transformation via free Bosons}
In the CKP case, the corresponding Darboux transformation must keep the CKP constraint, that is, $\pa^*L+\pa^*L^*=0$. The suitable basic CKP Darboux transformation \cite{Nimmo1995,He2006,Hewz2007,Loris1999} is just $T(q)\triangleq T(q,q)$, corresponding to the squared eigenfunction symmetry \cite{Cheng2014,Loris2001}
$\pa_{q,q}L=[q\pa^{-1}q,L]$,
which is the special case of $\pa_{(q_1,q_2)}L\triangleq[q_1\pa^{-1}q_2+q_2\pa^{-1}q_1, L]$
with $q_2=q_1/2=q/2$. Firstly,
by comparing the coefficient of $\pa^{-1}$ in $\pa^*_{q,q}W=q\pa^{-1}qW$, one can get
\begin{eqnarray*}
\pa^*_{q,q} \tau_0=\frac{1}{2}\Omega(q,q)\tau_0.
\end{eqnarray*}
Further by (\ref{paqrobj}) and induction, the general case is given in the following lemma.
\begin{lemma}\label{pamt0}
For $m\geq 1$,
\begin{eqnarray*}
\pa^{*m}_{q,q}(\tau_0(t))=\frac{1}{2^m}(2m-1)!!\Omega(q,q)^m\tau_0(t).
\end{eqnarray*}
\end{lemma}
Then one can obtain the lemma below
\begin{lemma}\label{tau0g1dtses}
Under the CKP Darboux transformation $T{(q)}$,
\begin{eqnarray*}
\tau_0(t)^{[1]}=e^{\pa^*_{q,q}}(\tau_0(t)),\quad g_1(t,z)^{[1]}=e^{\pa^*_{q,q}}(g_1(t,z)).
\end{eqnarray*}
\end{lemma}
\begin{proof}
Firstly for $w(t,\lambda)=(1+w_1\lambda^{-1}+w_2\lambda^{-2}+\cdots)e^{\xi(t,\lambda)}$, one can obtain $w_1^{[1]}=w_1-q^2/\Omega(q,q)$ by \eqref{omegaexpansion} and \eqref{q1dt}. Then by $w_1=2\pa_x\log \tau_0$, it can be found that $\tau_0^{[1]}=\tau_0/\sqrt{-\Omega(q,q)}$, since $\tau_0$ can be determined up to a multiplication of a constant. While according to Lemma \ref{pamt0},
\begin{align*}
e^{\pa^*_{q,q}}(\tau_0)=\sum_{m\geq 0}\frac{(-1/2)(-3/2)\cdots(1/2-m)}{m!}(-\Omega(q,q))^m\tau_0
=\tau_0(1-\Omega(q,q))^{-1/2}.
\end{align*}
So $\tau_0(t)^{[1]}=e^{\pa^*_{q,q}}(\tau_0(t))$ holds when viewing $\Omega(q,q)-1$ as $\Omega(q,q)$.

As for $g_1(t,z)^{[1]}=e^{\pa^*_{q,q}}(g_1(t,z))$, one can prove it by Proposition \ref{onedtsesy} and $w(t,\lambda)=2g_1(t,\lambda)/\tau_0(t)$.
\end{proof}
Next we will try to seek the expressions of $\tau_0(t)^{[1]}$ and $g_1(t,z)^{[1]}$ by using free Bosons. Note that
the expression of $\tau_0^{[1]}$ under $T(q)$ by free Bosons are given in the proposition below.
\begin{proposition}\label{proptau01}
$\tau^{[1]}_0(t)$ under $T(q)$ can be expressed by
\begin{eqnarray*}
\tau^{[1]}_0(t)=\langle 0|e^{H(t)} e^{\frac{\beta^2}{2}}g|0\rangle,
\end{eqnarray*}
where $\beta\in V=\sum_{l\in \mathbb{Z}+1/2}\mathbb{C}\phi_l$ is corresponding to the eigenfunction $q(t)$ of the CKP hierarchy in the way (\ref{qbosonexp}).
\end{proposition}
\begin{proof}
In fact by Lemma \ref{pamt0} and Lemma \ref{tau0g1dtses}, one only needs to prove
\begin{eqnarray}
\langle 0|e^{H(t)} \beta^{2m}g|0\rangle=(2m-1)!!\Omega(q,q)^m\tau_0(t).\label{2mproof}
\end{eqnarray}
Notice that it is right for $m=0$ and $m=1$ by Proposition \ref{omegaqqboson}. So we can assume it is correct for $<m$. In this case,  by using Lemma $\ref{lemmas}$ and applying $\langle1|e^{H(t)} \otimes\langle0|e^{H(t')}$,
\begin{eqnarray*}
&&{\rm Res}_{z}w(t,z)\frac{\langle0|e^{H(t')}\phi(-z)\beta^{2m-1} g|0\rangle)}{\langle0|e^{H(t')}g|0\rangle}=-(2m-1)q(t)\frac{\langle0|e^{H(t')}
\beta^{2m-2}g|0\rangle}{\langle0|e^{H(t')}g|0\rangle}.
\end{eqnarray*}
Further by using the assumption of $m-1$ and the spectral representation of $q(t)$ in \eqref{qspec},
\begin{eqnarray*}
&&{\rm Res}_{z}w(t,z)\left(\frac{\langle0|e^{H(t')}\phi(-z)\beta^{2m-1} g|0\rangle)}{\langle0|e^{H(t')}g|0\rangle}
-(2m-1)!!\Omega(q(t'),q(t'))^{m-1}\Omega(q_1(t'),w(t',-z))\right)=0.
\end{eqnarray*}
By using similar method in Proposition \ref{omegaqqboson}, one can find that $\langle0|e^{H(t')}\phi(-z)\beta^{2m-1} g|0\rangle)=\mathcal{O}(z^{-1})e^{-\xi(t',z)}$. Therefore by Lemma \ref{2wave} and spectral representation of $q$ in \eqref{qspec}, one can at last prove \eqref{2mproof} is correct for $m$.
\end{proof}
Now let us discuss the expression of $g_1^{[1]}(t,z)$ by free Bosons, which is given in the proposition below.
\begin{proposition}\label{propg1tz}
Under the same conditions of Proposition \ref{proptau01},
\begin{eqnarray*}
g_1^{[1]}(t,z)=\langle 1|e^{H(t)} \phi(z)e^{\frac{\beta^2}{2}}g|0\rangle.
\end{eqnarray*}
\end{proposition}
\begin{proof}
Similar to Proposition \ref{proptau01}, we only need to prove
\begin{eqnarray}
\pa^{*m}_{q,q}(g_1(t,z))=\frac{1}{2^m}\langle 1|e^{H(t)} \phi(z)\beta ^{2m}g|0\rangle.\label{g1boson}
\end{eqnarray}
Notice that \eqref{g1boson} is obviously true for $m=0$. Now we assume \eqref{g1boson} holds for $<m$. By using \eqref{paqrmq} and Lemma \ref{pamt0},
\begin{eqnarray*}
\pa^{*m}_{q,q}(g_1(t,z))=
\frac{1}{2}\pa^{*m}_{q,q}\Big(w(t,z)\tau_0(t)\Big)=
\frac{1}{2}\sum_{k=0}^m\frac{1}{2^{m-k}}k!(2m-2k-1)!!
C_{m}^{k}q\Omega(q,q)^{m-1}\Omega(q,w)\tau_0.
\end{eqnarray*}
So now we only need to prove
\begin{eqnarray}
\sum_{k=0}^m2^{k-1}k!(2m-2k-1)!!
C_{m}^{k}q\Omega(q,q)^{m-1}\Omega(q,w)\tau_0=\langle 1|e^{H(t)} \phi(z)\beta^{2m}g|0\rangle.\label{omegaphibeta2m}
\end{eqnarray}
In fact by Lemma $\ref{lemmas}$, one can obtain $S(1\otimes\beta ^{2m})(g|0\rangle\otimes g|0\rangle)=-2m\beta g|0\rangle\otimes\beta ^{2m-1}g|0\rangle$. Therefore
\begin{eqnarray*}
{\rm Res}_{z}w(t,z)\frac{\langle1|e^{H(t')}\phi(-z)\beta ^{2m} g|0\rangle}{\langle0|e^{H(t')}g|0\rangle}=-2mq(t)
\frac{\langle1|e^{H(t')}\beta ^{2m-1}g|0\rangle}{\langle0|e^{H(t')}g|0\rangle}.
\end{eqnarray*}
Further by bilinear equation \eqref{waveckpbilinear} and spectral representation \eqref{qspec} of $q(t)$,
\begin{align*}
{\rm Res}_{z}w(t,z)\Big(&\frac{\langle1|e^{H(t')}\phi(-z)\beta ^{2m} g|0\rangle}{\langle0|e^{H(t')}g|0\rangle}-\frac{1}{2^{m-1}}w(t',-z)\pa^{*m}_{q,q}(\tau_0(t'))
\\&-2m\Omega(q(t'),w(t',-z))\frac{\langle1|e^{H(t')}
\beta ^{2m-1}g|0\rangle}{\langle0|e^{H(t')}g|0\rangle}\Big)=0.
\end{align*}
It can be proved the terms in the bracket above satisfy the condition of Lemma  \ref{2wave}. Thus
\begin{eqnarray}
\frac{\langle1|e^{H(t)}\phi(z)\beta ^{2m} g|0\rangle}{\langle0|e^{H(t)}g|0\rangle}=\frac{1}{2^{m-1}}w(t,z)
\pa^{*m}_{q,q}(\tau_0(t))+2m
\Omega(q_1(t),w(t,z))\frac{\langle1|e^{H(t)}
\beta ^{2m-1}g|0\rangle}{\langle0|e^{H(t)}g|0\rangle}.\label{mgi}
\end{eqnarray}
So if insert the expression of $\langle 1|e^{H(t)} \beta^{2m-1}g|0\rangle$ (just replace the $w(t,z)$ in (\ref{omegaphibeta2m}) for $m-1$ with $q$) into the right hand side of \eqref{mgi}, one can at last prove (\ref{omegaphibeta2m}) is also correct for $m$.
\end{proof}
By Proposition \ref{fockiso}, Proposition \ref{proptau01} and Proposition \ref{propg1tz}, one can obtain the corollary below.
\begin{corollary}\label{corwavetrans}
Under the Darboux transformation $T(q)$, the wave function $w(t,z)$ of the CKP hierarchy can be expressed in free Bosons as
\begin{eqnarray*}
w(t,z)^{[1]}=\frac{2\langle1|e^{H(t)}\phi(z)e^{\frac{\beta^2}{2}} g|0\rangle}{\langle0|e^{H(t)}e^{\frac{\beta^2}{2}}g|0\rangle},
\end{eqnarray*}
where Bosonic field $\beta\in V=\sum_{l\in \mathbb{Z}+1/2}\mathbb{C}\phi_l$ is corresponding to $q$ by (\ref{qbosonexp}). In particular,
\begin{align*}
\tau_0^{[1]}(t)=
\langle0|e^{H(t)}e^{\frac{\beta^2}{2}}g|0\rangle,\quad
\tau_{(2j-1,1)}^{[1]}(t)=\langle(2j-1,1)|e^{H(t)}e^{\frac{\beta^2}{2}}g|0\rangle,\quad j=2,3,\cdots,
\end{align*}
where $\tau_0^{[1]}$ and $\tau_{(2j-1,1)}^{[1]}$ be the tau functions corresponding to $w(t,z)^{[1]}$.
\end{corollary}
\begin{conjecture}\label{conjecture}
Under the CKP Darboux transformation $T(q)$ with $\beta\in V=\sum_{l\in \mathbb{Z}+1/2}\mathbb{C}\phi_l$ corresponding to $q$, then
$$\tau=g|0\rangle\xrightarrow{T(q)}e^{\frac{\beta^2}{2}}\tau,$$
where $g\in Sp_{\infty}$.
\end{conjecture}
\noindent{\bf Remark:} The major difficulty of this conjecture is to determine the evolution equations of $w_{\alpha}(t,z)$. One can refer to the remark below Eq. (\ref{ckplaxeq}) for more details. But it does not affect the applications of free Bosons in CKP hierarchy.

Next we will consider the successive applications of CKP Darboux transformations in the language of free Bosons. For this, denote $\rho_1^{[1]}(z)$ and $\beta_1^{[1]}\in V$ in the way below,
\begin{align*}
\rho_1^{[1]}(z)=\Omega(q_1^{[1]}(t'),\psi(t',-z)^{[1]}),
\quad\beta_1^{[1]}=-{\rm Res}_{z}\phi(z)\rho_1^{[1]}(z).
\end{align*}
By direct computation, one can obtain the lemma below
\begin{lemma}\label{transep}
Under the CKP Darboux transformation $T(q)$,
\begin{align*}
\rho_1^{[1]}(z)=\Omega(q_1(t')^{[1]},\psi(t',-z)^{[1]})=
\Omega(q_1(t'),\psi(t',-z))
-\frac{\Omega(q_1(t'),q(t'))}{\Omega(q(t'),q(t'))}\Omega(q(t'),\psi(t',-z)),
\end{align*}
and
\begin{align*}
\beta_1^{[1]}=
\beta_1
-\frac{\Omega(q_1(t'),q(t'))}{\Omega(q(t'),q(t'))}\beta,
\end{align*}
where $\beta,\beta_1\in V=\sum_{l\in \mathbb{Z}+1/2}\mathbb{C}\phi_l$ are corresponding to eigenfunctions $q$ and $q_1$ according to (\ref{qbosonexp}).
\end{lemma}
\noindent {\bf Remark:} By Corollary \ref{corwavetrans}, one can found that $\phi(z)^{[1]}=\phi(z)$.

Notice that $q_1^{[1]}$ can be expressed by the spectral representation in terms of $w(t,z)^{[1]}$, therefore one can obtain the proposition below by Corollary \ref{corwavetrans}.
\begin{proposition}
Under the Darboux transformation $T(q)$, the eigenfunction $q_1(t)$ of the CKP hierarchy can be expressed in terms of free Bosons
\begin{eqnarray}
q_1(t)^{[1]}=\frac{2\langle1|e^{H(t)}\beta_1^{[1]}e^{\frac{\beta^2}{2}} g|0\rangle}{\langle0|e^{H(t)}e^{\frac{\beta^2}{2}}g|0\rangle},
\end{eqnarray}
where $\beta_1^{[1]}$ satisfies the relation in Lemma \ref{transep}.
\end{proposition}
\noindent {\bf Remark:} $\frac{\Omega(q_1(t'),q(t'))}{\Omega(q(t'),q(t'))}$ in the relation between $\beta_1^{[1]}$ and $\beta_1$ in Lemma \ref{transep} can be viewed as one constant independent of $t$. We can set it to be zero, for example, when $\beta_1=\phi(z)$.

If consider the
following CKP Darboux transformation chain,
\begin{eqnarray}
&&L\xrightarrow{T(q_1)}L^{[1]}\xrightarrow{T(q^{[1]}_2)}L^{[2]}
  \rightarrow\cdots\rightarrow L^{[k-1]}\xrightarrow{T(q^{[k-1]}_k )}L^{[k]},\label{chaint}
\end{eqnarray}
and denote
\begin{eqnarray*}
\rho_i^{[j]}(z)=\Omega(q_i^{[j]}(t'),\psi(t',-z)^{[j]}),\quad \beta_i^{[j]}=-{\rm Res}_z\phi(z)\rho_i^{[j]}(z)
\end{eqnarray*}
and $T^{[\vec{\bf k}]}=T(q_{k}^{[k-1]})\cdots T(q_1)$, then one can obtain the following proposition.
\begin{proposition}\label{ndttd}
Under the $k$-step Darboux transformation $T^{[\vec{\bf k}]}$ (see chain (\ref{chaint})),
\begin{align*}
w(t,z)^{[k]}=\frac{2\langle1|e^{H(t)}\phi(z)e^{\frac{(\beta^{[k-1]}_k)^2}{2}}\cdots e^{\frac{(\beta^{[1]}_1)^2}{2}}e^{\frac{\beta_1^2}{2}} g|0\rangle}{\langle0|e^{H(t)}e^{\frac{(\beta^{[k-1]}_k)^2}{2}}\cdots e^{\frac{(\beta^{[1]}_1)^2}{2}}e^{\frac{\beta_1^2}{2}}g|0\rangle},\\
q(t)^{[k]}=\frac{2\langle1|e^{H(t)}\beta^{[k]}e^{\frac{(\beta^{[k-1]}_k)^2}{2}}\cdots e^{\frac{(\beta^{[1]}_1)^2}{2}}e^{\frac{\beta_1^2}{2}} g|0\rangle}{\langle0|e^{H(t)}e^{\frac{(\beta^{[k-1]}_k)^2}{2}}\cdots e^{\frac{(\beta^{[1]}_1)^2}{2}}e^{\frac{\beta_1^2}{2}}g|0\rangle}.
\end{align*}
\end{proposition}

\section{Additional symmetries of CKP hierarchy}
In this section, we will discuss the actions of the additional symmetries on the tau functions for the CKP hierarchy. For this, let us firstly review some backgrounds on the additional symmetries of the KP hierarchy.

The additional symmetry flow $\pa^*_{ml}$ on the KP Lax operator $\mathfrak{L}$ are introduced by\cite{ASvM,Dickey1995,OS}
$$\pa^*_{ml}\mathfrak{L}=-[(\mathfrak{M}^m\mathfrak{L}^l)_-,\mathfrak{L}],$$
where $\mathfrak{M}=\mathfrak{D}\cdot(\sum_{n=1}nt_n\pa^{n-1})\cdot\mathfrak{D}^{-1}$ with the dressing operator $\mathfrak{D}$ is Orlov-Shulman (OS) operator \cite{OS}. If denote the generator of the additional symmetries as follows,
$$Y(\mathfrak{t},\lambda,\mu)=\sum_{m=0}^\infty\frac{(\mu-\lambda)^m}{m!}
\sum_{l=-\infty}^\infty\lambda^{-l-m-1}(\mathfrak{M}^m\mathfrak{L}^{m+l})_-,$$
then one can obtain the proposition below \cite{Dickey1995}.
\begin{proposition}\cite{Dickey1995}
Given the wave function $\psi(\mathfrak{t},\mu)$ and adjoint wave functions $\psi^*(\mathfrak{t},\lambda)$ of the KP hierarchy corresponding to the Lax operator $\mathfrak{L}$,
$$Y(\mathfrak{t},\lambda,\mu)=\psi(\mathfrak{t},\mu)\pa^{-1}\psi^*(\mathfrak{t},\lambda).$$
\end{proposition}
Therefore, one can find the the generator of the additional symmetries is just the squared eigenfunction symmetry \cite{Dickey1995,Aratyn}
$$\pa^*_{\psi(\mathfrak{t},\mu),\psi^*(\mathfrak{t},\lambda)}=-\sum_{m=0}^\infty\frac{(\mu-\lambda)^m}{m!}
\sum_{l=-\infty}^\infty\lambda^{-l-m-1}\pa^*_{m,m+l}.$$
The actions of the additional symmetries on the wave function $\psi(\mathfrak{t},z)$ and tau function $\tau_{\rm KP}(\mathfrak{t})$ are related by so-called famous Adler-Shiota-van Moerbeke-Dickey formula \cite{ASvM,Dickey1995}, that is,
$$\pa^*_{\psi(\mathfrak{t},\mu),\psi^*(\mathfrak{t},\lambda)}\psi(\mathfrak{t},z)
=\psi(\mathfrak{t},z)\Big(e^{\sum_{l=1}^\infty\frac{1}{l}z^{-1}\frac{\pa}{\pa t_l}}-1\Big)\frac{X(\mathfrak{t},\lambda,\mu)\tau_{\rm KP}(\mathfrak{t})}{\tau_{\rm KP}(\mathfrak{t})}.$$
This can be obtained from $$\pa^*_{\psi(\mathfrak{t},\mu),\psi^*(\mathfrak{t},\lambda)}
\tau_{\rm KP}(\mathfrak{t})=-
\Omega(\psi(\mathfrak{t},\mu),\psi^*(\mathfrak{t},\lambda))
\tau_{\rm KP}(\mathfrak{t})\quad\text{(see (\ref{paqrobj}))}$$
and the following expression of SEP \cite{Aratyn}
$$\Omega(\psi(\mathfrak{t},\mu),\psi^*(\mathfrak{t},\lambda))
=-\frac{X(\mathfrak{t},\lambda,\mu)\tau_{\rm KP}(\mathfrak{t})}{\tau_{\rm KP}(\mathfrak{t})}$$
with the vertex operator given by
\begin{eqnarray*}
X(\mathfrak{t},\lambda ,\mu )&=&\frac{1}{\lambda}e^{\xi(\mathfrak{t}+[\lambda^{-1}],\mu)
-\xi(\mathfrak{t},\lambda)}e^{\sum_{l=1}^\infty\frac{1}{l}(\lambda^{-1}-\mu^{-1})\frac{\pa}{\pa t_l}}\\
  &=&-\frac{1}{\mu}e^{\xi(\mathfrak{t},\mu)-\xi(\mathfrak{t}-[\mu^{-1}],\lambda)}e^{\sum_{l=1}^\infty\frac{1}{l}(\lambda^{-1}-\mu^{-1})\frac{\pa}{\pa t_l}}+\delta(\lambda,\mu).
\end{eqnarray*}
Here $\delta(\lambda,\mu)=\frac{1}{\lambda}\frac{1}{1-\mu/\lambda}+\frac{1}{\mu}\frac{1}{1-\lambda/\mu}$ has support $\lambda\neq \mu$, that is, we can believe that $\delta(\lambda,\mu)$ is zero when  $\lambda\neq \mu$.

As for the CKP hierarchy, the corresponding additional symmetry on the Lax operator $L$ \cite{He2007} is given by
\begin{align*}
\pa^*_{ml}L=-[(M^mL^l-(-1)^lL^lM^m)_-,L],
\end{align*}
where $M=W\cdot(\sum_{n\in 1+2\mathbb{Z}_{\geq 0}}nt_n\pa^{n-1})\cdot W^{-1}$ is the OS operator of the CKP hierarchy \cite{He2007}.
Then the corresponding generator of CKP additonal symmetries \cite{He2007}
$$Y^C(\lambda,\mu)=\sum_{m=0}^\infty\frac{(\mu-\lambda)^m}{m!}
\sum_{l=-\infty}^\infty\lambda^{-l-m-1}(M^mL^{m+l}-(-1)^{m+l}L^{m+l}M^m)_-,$$
are given by the following proposition.
\begin{proposition}\cite{He2007}
$$Y^C(\lambda,\mu)=w(t,\mu)\pa^{-1}w(t,-\lambda)+w(t,-\lambda)\pa^{-1}w(t,\mu).$$
\end{proposition}
Therefore
$$\pa^*_{w(t,\mu),w(t,-\lambda)}+\pa^*_{w(t,-\lambda),w(t,\mu)}=-\sum_{m=0}^\infty\frac{(\mu-\lambda)^m}{m!}
\sum_{l=-\infty}^\infty\lambda^{-l-m-1}\pa^*_{m,m+l}.$$
If set $\pa^*_{\lambda\mu}=\pa^*_{w(t,\mu),w(t,-\lambda)}+\pa^*_{w(t,-\lambda),w(t,\mu)}$ for short, then it can be found that
\begin{align}
\pa^*_{\lambda\mu}
=\pa^*_{w(t,\mu)+w(t,-\lambda),w(t,\mu)+w(t,-\lambda)}-\pa^*_{w(t,\mu),w(t,\mu)}-\pa^*_{w(t,-\lambda),w(t,-\lambda)},
\end{align}
that is,
\begin{align*}
\pa^*_{\lambda\mu}L=&
\Big[\big(w(t,\mu)+w(t,-\lambda))\pa^{-1}(w(t,\mu)+w(t,-\lambda)\big)\nonumber\\
&-w(t,\mu)\pa^{-1}w(t,\mu)-w(t,-\lambda)\pa^{-1}w(t,-\lambda),L\Big].
\end{align*}
Thus if denote $w(t,z)^{\{1\}}\triangleq e^{\pa^*_{\lambda\mu}}\Big(w(t,z)\Big)$, by using Lemma \ref{paqrcomm} and Corollary \ref{cordtkstep} one can find that
\begin{align*}
w(t,z)^{\{1\}}
=&
e^{-\pa^*_{w(t,\mu),w(t,\mu)}}\cdot
e^{-\pa^*_{w(t,-\lambda),w(t,-\lambda)}}\cdot
e^{\pa^*_{w(t,\mu)+w(t,-\lambda),w(t,\mu)+w(t,-\lambda)}}\Big(w(t,z)\Big)\nonumber\\
=&\frac{2\langle1|e^{H(t)}\phi(z)e^{-\frac{\phi(\mu)^2}{2}} e^{-\frac{\phi(-\lambda)^2}{2}}e^{\frac{(\phi(\mu)+\phi(-\lambda))^2}{2}} g|0\rangle}{\langle0|e^{H(t)}e^{-\frac{\phi(\mu)^2}{2}} e^{-\frac{\phi(-\lambda)^2}{2}}e^{\frac{(\phi(\mu)+\phi(-\lambda))^2}{2}}g|0\rangle},
\end{align*}
where we have used $\phi(\mu)^{[1]}=\phi(\mu)$ and $-\pa^*_{w(t,\mu),w(t,\mu)}=\pa^*_{\sqrt {-1}w(t,\mu),\sqrt {-1}w(t,\mu)}$

On the other hand, one can rewrite the commutation relation of free Bosons into
$$\phi(\mu)\phi(-\lambda)-\phi(-\lambda)\phi(\mu)=\delta(\lambda,\mu).$$
Therefore when $\mu\neq\lambda$, we can believe that $\phi(\mu)$ commute with $\phi(\lambda)$. Based upon this,
\begin{align*}
e^{-\frac{\phi(\mu)^2}{2}} e^{-\frac{\phi(-\lambda)^2}{2}}e^{\frac{(\phi(\mu)+\phi(-\lambda))^2}{2}}
=e^{\phi(\mu)\phi(-\lambda)}
\end{align*}
So one can at last obtain the proposition below.
\begin{proposition}
When $\mu\neq\lambda$,
\begin{align*}
w(t,z)^{\{1\}}=\frac{2\langle1|e^{H(t)}\phi(z)e^{\phi(\mu)\phi(-\lambda)} g|0\rangle}{\langle0|e^{H(t)}e^{\phi(\mu)\phi(-\lambda)}g|0\rangle},
\end{align*}
and
\begin{align*}
\tau_0^{\{1\}}(t)=
\langle0|e^{H(t)}e^{\phi(\mu)\phi(-\lambda)}g|0\rangle,\quad
\tau_{(2j-1,1)}^{\{1\}}(t)=\langle(2j-1,1)|e^{H(t)}e^{\phi(\mu)\phi(-\lambda)}g|0\rangle,\quad j=2,3,\cdots,
\end{align*}
where $\tau_0^{\{1\}}$ and $\tau_{(2j-1,1)}^{\{1\}}$ be the tau functions corresponding to $w(t,z)^{\{1\}}$.
Further
\begin{align*}
\pa^*_{\lambda\mu}\tau_0(t)
=\langle0|e^{H(t)}\phi(\mu)\phi(-\lambda)g|0\rangle,\quad
\pa^*_{\lambda\mu}\tau_{(2j-1,1)}(t)=\langle(2j-1,1)|e^{H(t)}\phi(\mu)\phi(-\lambda)g|0\rangle.
\end{align*}
\end{proposition}
\begin{proof}
The results of $w(t,z)^{\{1\}}$, $\tau_0^{\{1\}}$ and $\tau_{(2j-1,1)}^{\{1\}}$ can be obtained by Corollary \ref{corwavetrans} and Proposition \ref{ndttd}. While for $\pa^*_{\lambda\mu}\tau_0(t)$ and  $\pa^*_{\lambda\mu}\tau_{(2j-1,1)}(t)$, one can consider the the first order of $\epsilon$ expansion in $e^{\epsilon\pa^*_{\lambda\mu}}\Big(w(t,z)\Big)$ or compute directly by Lemma \ref{pamt0} and \eqref{g1boson}.
\end{proof}
Next we will compute the explicit forms of $\pa^*_{\lambda\mu}\tau_0(t)$ and $\pa^*_{\lambda\mu}\tau_{(2j-1,1)}(t)$.
\begin{proposition}\label{propaddontau}
The actions of $\pa^*_{\lambda\mu}$ on the CKP tau functions are given by
\begin{align*}
\pa^*_{\lambda\mu}\tau_0(t)=&e^{\xi(t,\mu)-\xi(t,\lambda)}
e^{\xi(\tilde{\pa},\mu^{-1})-\xi(\tilde{\pa},\lambda^{-1})}
\frac{\mu+\lambda}{\mu-\lambda}\\
&\times\left(\frac{\mu+\lambda}{(\mu-\lambda)^2}\tau_0(t)
+4\sum_{k\neq l}\mu^{-l}(-\lambda)^{-k}\tau_{(2k-1,2l-1)}(t)\right),\\
\pa^*_{\lambda\mu}\tau_{(2j-1,1)}(t)=&e^{\xi(t,\mu)-\xi(t,\lambda)}
e^{\xi(\tilde{\pa},\mu^{-1})-\xi(\tilde{\pa},\lambda^{-1})}
\frac{\mu+\lambda}{\mu-\lambda}\\
&\times\Big(4\sum_{k\neq l>0}\mu^{-l}(-\lambda)^{-k}\tau_{(2k-1,2l-1,2j-1,1)}(t)\\
&+\frac{\mu+\lambda}{(\mu-\lambda)^2}
\tau_{(2j-1,1)}(t)+(2j-1)\sum_{k>0}\big((-\mu)^{j-1}(-\lambda)^k
-\lambda^{j-1}\mu^k\big)\tau_{(2k-1,1)}(t)\\
&+\sum_{k>0}((-\lambda)^{-k}-\mu^{-k})\tau_{(2k-1,2j-1)}(t)
+\frac{1}{4}(2j-1)\left(\lambda^{j-1}-(-\mu)^{j-1}\right)\tau_0(t)\Big),
\end{align*}
where $k$ and $l$ in $\tau_{(2k-1,2l-1)}(t)$ (resp.  $\tau_{(2k-1,2l-1,2j-1,1)}(t)$) may not be $k>l$ (resp. $k>l>j$).
\end{proposition}
\begin{proof}
Let us firstly compute $\sigma\phi(\mu)\phi(-\lambda)\sigma^{-1}$. For this,
define the vertex operators $X_c(z)$ and $X_c(\lambda,\mu)$ as follows
\begin{eqnarray*}
X_c(z)\triangleq\sigma\phi(z)\sigma^{-1}=e^{\xi(t,z)}e^{\xi(\tilde{\pa},z^{-1})}A(t,z),\quad
X_c(\lambda,\mu)=X_c(\mu)X_c(-\lambda),
\end{eqnarray*}
where $\tilde{\partial}=(2\pa_{t_1},\frac{2\pa_{t_3}}{3},\cdots)$ and
\begin{eqnarray*}
A(t,z)&=&\sum_{0<j\in\mathbb{Z}}\Big(\frac{2j-1}{2}t_{\frac{2j-1}{2}}(-z)^{j-1}
+2\frac{\pa}{\pa t_{\frac{2j-1}{2}}}z^{-j} \Big).
\end{eqnarray*}
Therefore $\sigma\phi(\mu)\phi(-\lambda)\sigma^{-1}=X_c(\lambda,\mu)$. By using $e^{\xi(\tilde{\pa},\mu^{-1})} e^{-\xi(t,\lambda)}=\frac{\mu+\lambda}{\mu-\lambda}e^{-\xi(t,\lambda)}
e^{\xi(\tilde{\pa},\mu^{-1})}
$,
\begin{eqnarray*}
X_c(\lambda,\mu)&=&e^{\xi(t,\mu)}e^{\xi(\tilde{\pa},\mu^{-1})}A(t,\mu)e^{-\xi(t,\lambda)}
e^{-\xi(\tilde{\pa},\lambda^{-1})}A(t,-\lambda)\nonumber\\
&=&\frac{\mu+\lambda}{\mu-\lambda}e^{\xi(t,\mu)-\xi(t,\lambda)}
e^{\xi(\tilde{\pa},\mu^{-1})-\xi(\tilde{\pa},\lambda^{-1})}A(t,\mu)A(t,-\lambda).
\end{eqnarray*}
By computing $A(t,\mu)A(t,-\lambda)$, one can at last obtain
\begin{align*}
X_c(\lambda,\mu)&=\frac{\mu+\lambda}{\mu-\lambda}e^{\xi(t,\mu)-\xi(t,\lambda)}
e^{\xi(\tilde{\pa},\mu^{-1})-\xi(\tilde{\pa},\lambda^{-1})}\\
&\times\Big(\frac{\mu+\lambda}{(\mu-\lambda)^2}+\sum_{0<k\neq l\in\mathbb{Z}}\big(\frac{(2k-1)(2l-1)}{4}(-\mu)^{l-1}\lambda^{k-1}t_{\frac{2l-1}{2}}
t_{\frac{2k-1}{2}}\\
&+(2l-1)(-\mu)^{l-1}(-\lambda)^{-k}t_{\frac{2l-1}{2}}\frac{\pa}{\pa t_{\frac{2k-1}{2}}}+(2k-1)\mu^{-l}\lambda^{k-1}t_{\frac{2k-1}{2}}\frac{\pa}{\pa t_{\frac{2l-1}{2}}}\\
&+4 \mu^{-l} (-\lambda)^{-k}\frac{\pa}{\pa t_{\frac{2l-1}{2}}}\frac{\pa}{\pa t_{\frac{2k-1}{2}}}\big)\Big)
\end{align*}
Based upon this, one can obtain the result for $\pa^*_{\lambda\mu}\tau_0(t)$ by $\langle0|e^{H(t)}\phi(\mu)\phi(-\lambda)g|0\rangle=X_c(\lambda,\mu)\tau(\mathbf{t})\biggr|_{t_\mathrm{odd}=0}$ with $\tau(\mathbf{t})=\sigma(g|0\rangle)$.

As for $\pa^*_{\lambda\mu}\tau_{(2j-1,1)}(t)$, one needs to compute $\left[\frac{\pa}{\pa t_{\frac{2j-1}{2}}}\frac{\pa}{\pa t_{\frac{1}{2}}},X_c(\lambda,\mu)\right]$. Since $\frac{\pa}{\pa t_{\frac{2j-1}{2}}}\frac{\pa}{\pa t_{\frac{1}{2}}}$ commute with $\xi(t,\mu)$ and $\xi(\tilde{\pa},\mu^{-1})$, the key is to know
$\left[\frac{\pa}{\pa t_{\frac{2j-1}{2}}}\frac{\pa}{\pa t_{\frac{1}{2}}},A(t,\mu)A(t,-\lambda)\right]$. In fact, this can be seen from
\begin{align}
\frac{\pa}{\pa t_{\frac{2j-1}{2}}}A(t,\mu)+A(t,\mu)\frac{\pa}{\pa t_{\frac{2j-1}{2}}}=\frac{2j-1}{2}(-\mu)^{j-1}.
\end{align}
Further by $[AB,CD]=A[B,C]_+D-AC[B,D]_++[A,C]_+DB-C[A,D]_+B$ with $[A,B]_+=AB+BA$, one can obtain
\begin{align*}
\left[\frac{\pa}{\pa t_{\frac{2j-1}{2}}}\frac{\pa}{\pa t_{\frac{1}{2}}},A(t,\mu)A(t,-\lambda)\right]&=\frac{1}{2}A(t,-\lambda)\frac{\pa}{\pa t_{\frac{2j-1}{2}}}-\frac{1}{2}A(t,\mu)\frac{\pa}{\pa t_{\frac{2j-1}{2}}}
+\frac{1}{2}(\frac{2j-1}{2})\lambda^{j-1}\\&-\frac{1}{2}(\frac{2j-1}{2})(-\mu)^{j-1}
+
\frac{2j-1}{2}(-\mu)^{j-1}A(t,-\lambda)\frac{\pa}{\pa t_{\frac{1}{2}}}\\
&-\frac{2j-1}{2}\lambda^{j-1}A(t,\mu)\frac{\pa}{\pa t_{\frac{1}{2}}}.
\end{align*}
Therefore
\begin{align*}
\frac{\pa}{\pa t_{\frac{2j-1}{2}}}\frac{\pa}{\pa t_{\frac{1}{2}}}X_c(\lambda,\mu)
=&X_c(\lambda,\mu)\frac{\pa}{\pa t_{\frac{2j-1}{2}}}\frac{\pa}{\pa t_{\frac{1}{2}}}+\frac{\mu+\lambda}{\mu-\lambda}e^{\xi(t,\mu)-\xi(t,\lambda)}
e^{\xi(\tilde{\pa},\mu^{-1})-\xi(\tilde{\pa},\lambda^{-1})}\\
&\times\left((2j-1)\sum_{k>0}\Big((-\mu)^{j-1})(-\lambda)^{-k}-\lambda^{j-1}\mu^{-k}\Big)\frac{\pa}{\pa t_{\frac{2j-1}{2}}}\frac{\pa}{\pa t_{\frac{1}{2}}}\right.\\
&\left.+\sum_{k>0}\Big((-\lambda)^{-k}-\mu^{-k}\Big)\frac{\pa}{\pa t_{\frac{2k-1}{2}}}\frac{\pa}{\pa t_{\frac{2j-1}{2}}}+\frac{2j-1}{4}\Big(\lambda^{j-1}-(-\mu)^{j-1}\Big)\right)\\
&+\text{terms depending on $t_{\rm odd}$}.
\end{align*}
So the result of $\pa^*_{\lambda\mu}\tau_{(2j-1,1)}(t)$ can be derived by (\ref{taualpha}).

\end{proof}
\noindent{\bf Remark:}
Notice that the Adler-Shiota-van Moerbeke-Dickey formula in the CKP hierarchy can be obtained by
\begin{align*}
\pa^*_{\lambda\mu}w(t,z)=w(t,z)\left(e^{\xi(\tilde{\pa},z^{-1})}
\frac{\frac{1}{2}\pa^*_{\lambda\mu}\tau_0(t) +2 \sum_{\nu\in 1+
2\mathbb{Z},\, \nu>1} \pa^*_{\lambda\mu}\tau_{(\nu,1)}(t)z^{-\frac{\nu +1}{2}}}{\frac{1}{2}\tau_0(t) +2 \sum_{\nu\in 1+
2\mathbb{Z},\, \nu>1} \tau_{(\nu,1)}(t)z^{-\frac{\nu +1}{2}}}-\frac{\pa^*_{\lambda\mu}\tau_0(t)}{\tau_0(t)}\right).
\end{align*}
\section{Conclusions and Discussions}
By using free Bosons, we have obtained many important results of the CKP hierarchy. Firstly, the bilinear equations and the Lax equations of the modified CKP hierarchy are given in Proposition \ref{propmckp}. Then in Proposition \ref{propsolution}, the solutions of the constrained CKP hierarchy are derived in terms of free Bosons. Next the squared eigenfunction symmetry is showed to be the infinitesimal generator of the binary Darboux transformation in Proposition \ref{onedtsesy}. Based upon this, the CKP Darboux transformation is expressed in terms of free Bosons, which is given in Proposition \ref{proptau01}, Proposition \ref{propg1tz} and Corollary \ref{corwavetrans}. Since the generator of the additional symmetries is the squared eigenfunction symmetry, the CKP additional symmetries can be investigated by the CKP Darboux transformation. Finally we give the actions of the CKP additional symmetries on the CKP tau functions constructed by free Bosons in Proposition \ref{propaddontau}.

The results here are hoped to be helpful for better understanding the essential properties of the CKP hierarchy. In this paper, there is a conjecture (see Conjecture \ref{conjecture}) left. Though it does not affect the whole result of this paper, we are very interested in solving it. We believe the suitable way is to find the evolution equations of the wave function $w_\alpha(t,z)$ firstly. To do this, something else involving free Bosons must be needed. Therefore, we will concentrate on studying this question in future.


\begin{thebibliography}{99}
\bibitem{ASvM}M. Adler, T. Shiota and P. van Moerbeke, A Lax representation for the vertex operator and the central extension, Comm. Math. Phys. 171 (1995) 547-588.

\bibitem{Anguelova2017}I. I. Anguelova, The second bosonization of the CKP hierarchy, J. Math. Phys. 58 (2017) 071707.

\bibitem {Anguelova2018}I. I. Anguelova, The two bosonizations of the CKP hierarchy: overview and character identities, in {\it Representations of Lie algebras, quantum groups and related topics}, 1-34, Contemp. Math., 713, Amer. Math. Soc., Providence, RI, 2018 (arXiv:1708.04992).

\bibitem{Aratyn}H. Aratyn, E. Nissimov and  S. Pacheva, Method of squared eigenfunction potentials in integrable hierarchies of KP type, Comm. Math. Phys. 193 (1998) 493-525.

\bibitem{Wu2013}L. Chang and C. Z. Wu, Tau function of the CKP hierarchy and nonlinearizable Virasoro symmetries, Nonlinearity 26 (2013) 2577-2596.


\bibitem{chang}X. K. Chang, X. B. Hu and S. H. Li, Degasperis-Procesi peakon dynamical system and finite Toda lattice of CKP type, Nonlinearity 31 (2018) 4746-4775.

\bibitem{Chau1992}L. L. Chau, J. C. Shaw and H. C. Yen, Solving the KP hierarchy by gauge transformations, Comm. Math. Phys. 149 (1992) 263-278.

\bibitem{CLL}H. H. Chen, Y. C. Lee and J. E. Lin,
On a new hierarchy of symmetries for the Kadomtsev-Petviashvili equation,
Phys. D 9 (1983) 439-445.

\bibitem{Cheng2014}J. P. Cheng and J. S. He, The ``ghost" symmetry in the CKP hierarchy, J. Geom. Phys. 80 (2014) 49-57.

\bibitem{Cheng2021}J. P. Cheng and T. Milanov, The extended D-Toda hierarchy, Selecta Math. (N.S.) 27 (2021) 24,  85 pages.

\bibitem{Jimbo}M. Jimbo and T. Miwa, Solitons and infinite dimensional Lie algebras, Publ. RIMS, Kyoto Univ.19 (1983) 943-1001.

\bibitem{Date19814} E. Date, M. Jimbo, M. Kashiwara and T. Miwa, Transformation groups for soliton equations. VI. KP hierarchies of orthogonal and symplectic type, J. Phys. Soc. Japan. 50(1981) 3813-3818.

\bibitem{DJKM}E. Date, M. Kashiwara, M. Jimbo and T. Miwa, Transformation groups for soliton equations, in {\it Nonlinear integrable systems¡ªclassical theory and quantum theory} (Kyoto, 1981), 39-119, World Sci. Publishing, Singapore, 1983.

\bibitem{Dickey1995}L. A. Dickey, On additional symmetries of the KP hierarchy and Sato's B\"acklund transformation, Comm. Math. Phys. 167 (1995) 227-233.
\bibitem{feng}B. F. Feng, K. Maruno and Y. Ohta, The Degasperis-Procesi equation, its short wave model and the CKP hierarchy, Ann. Math. Sci. Appl. 2 (2017) 285-316.

\bibitem{FF}A. S. Fokas and B. Fuchssteiner, On the structure of symplectic operators and hereditary symmetries, Lett. Nuovo Cimento 28 (1980) 299-303.
\bibitem{fuwei}W. Fu and Frank W. Nijhof, On non-autonomous differential-difference AKP, BKP and CKP equations, Proc. A. 477 (2021) 1-20.

\bibitem{Geng2019}L. M. Geng, H. Z. Chen, N. Li and J. P. Cheng, Bilinear identities and squared eigenfunction symmetries of the BCr-KP hierarchy, J. Nonlinear Math. Phys. 26 (2019) 404-419.

\bibitem{He2006}J. S. He, Y. Cheng and R. A. Roemer, Solving bi-directional soliton equations in the KP hierarchy by gauge transformation, J. High Energy Phys. 3 (2006) 103.


\bibitem{He2007}J. S. He, K. L. Tian, A. Foerster and W. X.  Ma, Additional symmetries and string equation of the CKP hierarchy, Lett. Math. Phys. 81 (2007) 119-134.

\bibitem{Hewz2007}J. S. He, Z. W. Wu and Y. Cheng, Gauge transformations for the constrained CKP and BKP hierarchies, J. Math. Phys. 48 (2007) 113519.

\bibitem{Krichever}I. Krichever and A. Zabrodin, Kadomtsev-Petviashvili turning points and CKP hierarchy, arXiv:2012.04482.

\bibitem{Li2020}C. Z. Li and R. L. Ge, Symmetries of supersymmetric CKP hierarchy and its reduction, J. Geom. Phys. 158 (2020) 103894.

\bibitem{licx}C. X. Li and S. H. Li, The Cauchy two-matrix model, C-Toda lattice and CKP hierarchy, J. Nonlinear Sci. 29 (2019) 3-27.

\bibitem{Liu2017}Q. F. Liu and C. Z. Li, Quantum torus symmetries of the CKP and multi-component CKP hierarchies,
J. Math. Phys. 58 (2017) 113505.
\bibitem{Liusq}S. Q. Liu, C. Z. Wu and Y. J. Zhang, On the Drinfeld-Sokolov hierarchies of D type,  Int. Math. Res. Not. IMRN 2011 (2011) 1952¨C1996.

\bibitem{Loris1999}I. Loris, On reduced CKP equations, Inverse Problems 15 (1999) 1099-1109.
\bibitem{Loris2001}I. Loris, Dimensional reductions of BKP and CKP hierarchies, J. Phys. A 34 (2001) 3447-3459.

\bibitem{Matveev1991} V. B. Matveev and M. A. Salle, Darboux transformations and solitons, Springer-Verlag, Berlin, 1991.

\bibitem{Miwa2000}T. Miwa, M.Jimbo and E. Date, Solitons: Differential equations, symmetries and infinite-dimensional algebras. Translated from the 1993 Japanese original by Miles Reid, Cambridge Tracts in Mathematics, 135, Cambridge University Press, Cambridge, 2000.

\bibitem{Nimmo1995}J. J. Nimmo, Darboux transformation from reduction of the KP hierarchy, in {\it Nonlinear evolution equation and dynamical systems}, ed. by V.G. Makhankov et al, World Scientific, Singapore, 1995, 168-177.

\bibitem{Oevelpa1993}W. Oevel, Darboux theorems and Wronskian formulas for integrable system I: constrained KP flows, Physica A 195 (1993) 533-576.

\bibitem{Oevelrmp1993}W. Oevel and C. Rogers, Gauge transformations and reciprocal links in 2+1 dimensions, Rev. Math. Phys. 5 (1993) 299-330.

\bibitem{Oevel1998} W. Oevel and S. Carillo, Squared eigenfunction symmetries for soliton equations. I, II, J. Math. Anal. Appl. 217 (1998) 161-178, 179-199.

\bibitem{Shaw1997}J. C. Shaw and M. H. Tu, Miura and auto-B\"acklund transformations for the cKP and cmKP hierarchies, J. Math. Phys. 38 (1997) 5756-5773.
\bibitem{Takasaki}T. Ikeda and K. Takasaki, Toroidal Lie algebras and Bogoyavlensky's (2+1)-dimensional equation,
Int. Math. Res. Not. IMRN 2001 (2001)  329¨C369.
\bibitem{tian}K. L. Tian, J. S. He, J. P. Cheng and Y. Cheng, Additional symmetries of constrained CKP and BKP hierarchies, Sci. China Math. 54 (2011) 257-268.
\bibitem{uenotakasaki}
K. Ueno and K. Takasaki, Toda lattice hierarchy, in {\it Group representations and systems of differential equations} (Tokyo, 1982), 1-95, Adv. Stud. Pure Math., 4, North-Holland, Amsterdam, 1984.
\bibitem{van2012}J. W. van de Leur, A. Y. Orlov and T. Shiota, CKP hierarchy, bosonic tau function and bosonization formulae, SIGMA Symmetry Integrability Geom. Methods Appl. 8 (2012) 28.

\bibitem{Willox}R. Willox and J. Satsuma, Sato theory and transformation groups. A unified approach to integrable systems, in {\it Discrete integrable systems}, 17-55, Lecture Notes in Phys., 644, Springer, Berlin, 2004.

\bibitem{Willox1998}R. Willox, T. Tokihiro, I. Loris and J. Satsuma, The fermionic approach to Darboux transformations, Inverse Problems 14 (1998) 745-762.




\bibitem{Yangc2021}Y. Yang and J. P. Cheng, Bilinear equations in Darboux transformations by Boson-Fermion correspondence, arXiv:2101.02520.


\bibitem{OS}A. Yu. Orlov and E. I. Schulman, Additional symmetries for integrable equations and conformal algebra representation, Lett. Math. Phys. 12 (1986) 171-179.

\bibitem{Zuo2014}D. F. Zuo, L. Zhang and Q. Chen, On the sub-KP hierarchy and its constraints, revisited,  Rev. Math. Phys. 26 (2014) 1450019.
























\end{thebibliography}
\end{document}